\documentclass[12pt]{iopart}
\usepackage[english]{babel}
\usepackage{mathrsfs}

\expandafter\let\csname equation*\endcsname\relax
\expandafter\let\csname endequation*\endcsname\relax

\usepackage{amsmath}
\usepackage{amsthm}
\usepackage{amsfonts}
\usepackage{BOONDOX-cal}
\usepackage{graphicx}
\usepackage{braket}
\usepackage{subcaption}
\usepackage{multirow}
\usepackage[table]{xcolor}
\usepackage{bbm}

\usepackage{wrapfig}

\definecolor{nicepurple}{HTML}{97a0cf}
\definecolor{lightgreen}{HTML}{519E8A}
\definecolor{lightergreen}{HTML}{A7E8BD}
\definecolor{melon}{HTML}{FCBCB8}
\definecolor{puce}{HTML}{C08497}

\newcommand{\one}{\mathbbm{1}}
\renewcommand{\newline}{\\\\\noindent}

\newcommand{\hilb}{\mathscr{H}}
\newcommand{\quotes}[1]{``#1"}
\newcommand{\lie}[1]{\mathfrak{#1}}
\newcommand{\extd}{\textrm{d}}
\newcommand{\Uone}{\textrm{U}(1)}
\newcommand{\R}{\mathbb{R}}
\newcommand{\Z}{\mathbb{Z}}

\DeclareMathOperator{\sign}{sign}
\DeclareMathOperator{\re}{Re}
\DeclareMathOperator{\im}{Im}
\newcommand{\Uqone}{\textrm{U}_q(1)}
\newcommand{\scpr}[2]{\langle#1\, \vert \, #2 \rangle}

\newcommand{\jmax}{m_{\text{max}}}
\newcommand{\expect}[1]{{\langle #1\rangle}}

\newtheorem*{lemmaApp}{Lemma C.1}

\begin{document}

\title[]{Towards quantum gravity with neural networks: Solving the quantum Hamilton constraint of U(1) BF theory}

\author{Hanno Sahlmann $^1$, Waleed Sherif $^2$\footnote{Author to whom any correspondence should be addressed}}

\address{Institute for Quantum Gravity, Department of Physics, Friedrich-Alexander-Universit\"{a}t Erlangen-N\"{u}rnberg (FAU), Staudtstraße 7, 91058 Erlangen, Germany}
\ead{$^1$ hanno.sahlmann@fau.de, $^2$ waleed.sherif@gravity.fau.de}
\vspace{10pt}
\begin{indented}
\item[] \today 
\end{indented}

\begin{abstract}
In the canonical approach of loop quantum gravity, arguably the most important outstanding problem is finding and interpreting solutions to the Hamiltonian constraint. In this work, we demonstrate that methods of machine learning are in principle applicable to this problem. We consider U(1) BF theory in 3 dimensions, quantized with loop quantum gravity methods. In particular, we formulate a master constraint corresponding to Hamilton and Gau{\ss} constraints using loop quantum gravity methods. To make the problem amenable for numerical simulation we fix a graph and introduce a cutoff on the kinematical degrees of freedom, effectively considering $\Uqone$ BF theory at a root of unity. We show that the Neural Network Quantum State (NNQS) ansatz can be used to numerically solve the constraints efficiently and accurately. We compute expectation values and fluctuations of certain observables and compare them with exact results or exact numerical methods where possible. We also study the dependence on the cutoff. 
\end{abstract}

%
%
%
%
%

\section{Introduction}
Gravity as described by Einstein's theory of general relativity (GR) is the given as the interplay between the geometry of a 4-dimensional spacetime and the matter within \cite{Wald:1984rg}. Finding a quantum theory for gravity is one of the big outstanding problems of theoretical physics. The problem is difficult because of the deep conceptual questions that have to be answered when space-time geometry becomes quantum. It is also difficult because the dynamics of gravity is difficult to solve. This is true already in the classical theory and indeed much more so in all the approaches to quantization. Hence research has often focused on reduced or approximate models \cite{Bondi:1947fta,Weinberg:1972kfs,Carroll:2004st}. Another type of simplification is reducing the number of spatial dimensions in the theory by considering 2+1-dimensional gravity \cite{Carlip:1995zj,Carlip:1995qv,Carlip:1998uc}. 
\newline
Loop quantum gravity (LQG) \cite{Rovelli:1997yv,Thiemann:2001gmi,Thiemann:2007pyv,Ashtekar:2004eh} is based on a gauge theoretical formulation of Einstein's theory. Quantized excitations of the gravitational field are one-dimensional (canonical LQG) or two-dimensional (spinfoam-LQG) and carry group representations as labels. In the canonical approach to LQG, Einstein's equations take the form of constraint equations. Arguably the most difficult task is solving --  and interpreting the solutions of -- the constraints, in particular for the Hamilton constraint \cite{Thiemann:1996aw,Thiemann:1996av,Reisenberger:1996pu,Varadarajan:2022dgg}. Therefore, this problem is often investigated in reduced and approximate models, in particular in cosmological scenarios \cite{Bojowald:2001xe,Ashtekar:2006wn}.
\newline
In the present work, we aim to establish an approach to numerical canonical quantum gravity using modern numerical methods which harness the power of machine learning techniques. The ultimate goal is to find approximate solutions to the Hamilton constraint of LQG and to investigate their properties. This is complementary to other numerical work in LQG, see for example  \cite{Bahr:2016hwc,Han:2020npv,Cunningham:2020uco,Dona:2022yyn}, in that it addresses the canonical theory, and in that it makes use of different techniques. In contrast to numerical work on loop quantum cosmology, it does not make use of classical symmetry reduction.  
\newline
Since machine learning techniques have not been applied to LQG as far as we know, and the full theory with its SU(2) gauge symmetry is full of technical difficulties, in the present work we will consider a simpler system, with a similar kinematical structure and a dynamics that is fully understood. The simplest starting point to LQG would be gravity in 2+1-dimensions which can be understood as a BF-theory \cite{Oda:1990rw}. Such a theory has simple expression for its constraints. A BF-theory is defined by an action of the type \cite{Celada:2016jdt}
\begin{equation}
\label{eq:BF_action}
    S[\omega, B] = \int_{M^{(n)}}\tr(B \wedge F(\omega)),
\end{equation}
for an $n$-dimensional manifold $M^{(n)}$ and a fixed compact Lie group $G$ equipped with a non-degenerate invariant bilinear form and a corresponding Lie algebra $\lie{g}$. Here, $B$ is a $\lie{g}$-valued $n-2$ form and $\omega$ is a $\lie{g}$-valued connection with a curvature 2-form $F(\omega) = \extd \omega + \omega \wedge \omega$ of the connection $\omega$. The stationary points of the action fulfill the equations
\begin{equation}
     \label{eq:bfConstraints}
     F(\omega) = 0 \quad , \quad \extd^{(\omega)}B = 0.
\end{equation}
The solutions of these constraints are therefore flat connections $\omega$ and covariantly constant fields $B$. The theory has no local degrees of freedom, making it only non-trivial in the presence of boundaries of and/or non-trivial cycles in $M^{(n)}$. 
\newline
For 3+1 gravity in the weak coupling limit, one obtains a gauge group of $\Uone^3$ \cite{Smolin:1992wj,Bakhoda:2022rut}. Motivated by this, we consider the simplest BF-theory where $M^{(3)} = \R \times M^{(2)}$ with $M^{(2)} = \R^2$ and $G = \Uone$. 
\newline
Clearly, to treat a system with infinitely many degrees of freedom on a computer, some form of truncation is necessary. In the present work, we truncate U(1) BF-theory in two ways: We restrict consideration to the degrees of freedom associated with a single graph $\gamma$ in $M^{(3)}$, a truncation similar to that used in lattice gauge theory, which has also been advocated in LQG in a different context \cite{Giesel:2006uj,Giesel:2006uk}. We will also truncate the functionals of $\omega$ by restricting to a finite list of irreducible U(1) representations. Viewed from another angle, we are actually working in a kinematical setup for $\Uqone$-BF theory with $q$ a root of unity. 
\newline
With these radical truncations, the kinematical Hilbert space for the theory becomes finite dimensional, yet huge practical problems remain. First and foremost is the exponential growth of the Hilbert space dimension with the size of the graph and the cutoff on the representations. This growth makes it impossible to naively deal with state vectors for any but the smallest graphs and lowest cutoffs, as no computer system has enough memory to do so. In the present work, we deal with this problem by using neural network quantum states (NNQS) \cite{Carleo:2017nvk} as an ansatz for states solving the quantum constraints. In a prosaic way, these NNQS can thus be seen as a parametrization of the physical quantum states with a number of parameters that is much easier to handle. In a more poetic way, they could be described as an artificial brain that is optimized to reproduce physical states.  
\newline
The use of NNQS may be compared to using tensor network (TN) methods for approximating ground states in the sense that both ansätze parameterise the quantum state in such a way that the coefficients can be reproduced by a fewer number of parameters. In tensor networks, the network topology is carefully chosen to reproduce a desired symmetry and  entanglement structure of the state. It is fixed, and typically quite local. In NNQS on the other hand, the network topology is quite non-local and effectively variable. The optimization will determine which connections will carry a non-zero weight. Also, a tensor network state is polynomial in the tensor coefficients, whereas the weights of the NNQS determine the state in variable non-linear ways.

\subsection{Quantum $\Uone$ and $\Uqone$ BF-theory}
The classical theory \eqref{eq:BF_action} can be quantised using LQG methods \cite{Baez:1999sr,Dittrich:2014wpa,Drobinski:2017kfm}, see also \cite{Corichi:1997us,Sahlmann:2002xv}. In LQG, one applies the Dirac quantisation algorithm to the GR written in terms of the Ashtekar-Barbero connection $A^I_a$ and the densitised triads $E^I_a$ \cite{Ashtekar:2004eh,Thiemann:2001gmi}. In that formalism, the observables which are well defined are akin to what one obtains in lattice gauge theory: holonomies which are the parallel transport of the connection along paths in the spatial manifold and fluxes which are the smearing of the densitised triads along 2-surfaces. In $\Uone$ BF-theory, analogous observables would be 
\begin{equation}
    h_c[\omega] = \mathcal{P}\exp\left(-\int_c \omega\right)  \quad , \quad B_c[B] = \int_c B,
\end{equation}
where $c$ ranges over the set of oriented paths in $M^{(2)}$, $h_c$ are $\Uone$-valued and $B_c$ are the line integrals of the canonical momentum $B$ on $M^{(2)}$. Such observables have a Poisson bracket of the form
\begin{equation}
\label{eq:commutator}
    \lbrace h_c, B_{c'} \rbrace = -I(c, c') h_c,
\end{equation}
where $I(c, c') \in \Z$ is the signed intersection number
\begin{equation}
    I(c, c') = \sum_{p \in e \cap c'} \sign (\epsilon_{ab} t^a_p t'^{b}_p),
\end{equation}
of the paths $c$ and $c'$. The sum is over all transversal isolated intersections of the two paths and $t_p , t_p'$ are the tangent vectors to $e$ and $e'$ respectively at $p$. Further, holonomies transform under gauge transformations only at their endpoints \cite{Bodendorfer:2016uat}. Namely, for the case of $\Uone$, then 
\begin{equation}
    h_c \mapsto e^{\lambda(s(c)) - \lambda(t(c))} h_c,
\end{equation}
where $s(c), t(c)$ denote the beginning and endpoints of the path $c$. Upon quantisation, one obtains the operators $\hat{h}_c , \hat{B}_{c}$ corresponding to these observables. These operators satisfy the appropriate canonical commutation relations such that
\begin{equation}
    [\hat{h}_c , \hat{B}_{c'}] = -i\hbar I(c, c')\hat{h}_c.
\end{equation}
Such relations can be realised on a Hilbert space of the form (\cite{Drobinski:2017kfm}, see also \cite{Dittrich:2014wpa})
\begin{equation}
    \hilb = L^2(\mathcal{A}, \extd \mu),
\end{equation}
where $\mathcal{A}$ is a space of distributional connections $A$ and $\mu$ a measure on that space \cite{Thiemann:2001gmi,Ashtekar:2004eh,Bodendorfer:2016uat}. A dense subset in $\hilb$ is given by cylindrical functions, that is, complex valued functions which depend on the connection $A$ through a finite number of holonomies along paths
\begin{equation}
\label{eq:cylindrical}
    \Psi(A) = \psi (h_{e_1}(A), \dots , h_{e_n}(A)).
\end{equation}
Here, $\psi : \Uone^n \rightarrow \mathbb{C}$. The paths $e_1, \dots , e_{n}$ only intersect at endpoints, which we denote as vertices $v$. All together, they are said to constitute a graph $\gamma$ with a set of edges $e_k \in E(\gamma)$ and vertices $v \in V(\gamma)$. 
\newline
The functions of the form \eqref{eq:cylindrical} can be expanded into irreducible representations of $\Uone^n$. In this expansion, each edge then carries one irreducible representation of $\Uone$, labeled by a charge $m \in \Z$. We will call elements of this basis $\{ \Psi_{\underline{m}} \}$ for the cylindrical functions charge networks. We say that in the charge network state $\{\Psi_{\underline{m}} \}$ with $\underline{m}=(m_{e_1},\ldots, m_{e_n})$, edge $e_i$ has charge $m_{e_i}$.
\newline
Furthermore, one has for every vertex an intertwiner which intertwines the tensor product of the holonomy representations of the incoming and outgoing edges at the vertices accordingly. For this simple case of $\Uone$, this is just a Kronecker $\delta$, requiring that all the incoming and outgoing charges on the attached edges sum to zero \cite{Thiemann:2021hpa}.
\newline
Following this prescription, one obtains a kinematical Hilbert space for $\Uone$ BF-theory. Holonomies are defined as multiplication operators on cylindrical functions, and flux operators by
\begin{equation}
    B_c\, \psi := [B_c,\psi],
\end{equation}
where $\psi$ is understood as a multiplication operator on the right and as a state on the left. The commutator on the right is defined by \eqref{eq:commutator} since a general $\Psi$ can be written as a sum of products of holonomies. Defined in this way, the fluxes are not selfadjoint for a general measure $\text{d}\mu$, but for some measures they can be complemented by the divergence of the measure in such a way that they become self adjoint without changing their commutation relations \cite{Sahlmann:2002xv}. 
\newline
An especially natural choice of measure $\mu$ is the Ashtekar-Lewandowski measure $\mu_\text{AL}$. Charge network states form an orthonormal basis for 
\begin{equation}
   \hilb_\text{AL} = L^2(\mathcal{A}, \extd \mu_\text{AL}),
\end{equation}
and the $B_c$ are selfadjoint. 
\newline
The analog of the classical constraints given in equation \eqref{eq:bfConstraints} would be two sets of constraint operators. The first enforces flatness of the holonomies and the second imposes $\Uone$ gauge invariance. For the case of the latter, this would mean that for every vertex in the graph, the sum of all the charges of the edges attached to it amounts to zero \cite{Bakhoda:2022rut}.  
\newline
The flatness constraint is slightly more difficult to implement. The solution turns out not to be normalizable in $ \hilb_\text{AL}$. Instead, it forms a different measure $\mu_\text{flat}$ -- the $\delta$-measure on flat connections \cite{Dittrich:2014wpa,Drobinski:2017kfm,Ashtekar:1994mh,Sahlmann:2011xu}. We will not describe this measure in detail but we will describe its pull back to the cylindrical functions on a single, contractible loop $\alpha$, in which this measure would be given by
\begin{equation}
    \int \Psi_\alpha \extd \mu_\text{flat} = \psi_\alpha(\one) = \int_{\Uone} \psi_\alpha(g) f_\alpha(g)  \extd \mu_0(g) =  \int \Psi_\alpha F_\alpha \extd \mu_\text{AL},
\end{equation}
with 
\begin{equation}
\label{eq:delta}
    f_\alpha(g)= \sum_{n=-\infty}^\infty g^n.
\end{equation}
Thus one can clearly see that the measure providing solutions of the flatness constraint is not absolutely continuous with respect to the Ashtekar-Lewandowski measure. In the following, we will nevertheless work with states in $\hilb_\text{AL}$. We will show that the solution \eqref{eq:delta} can be meaningfully approximated by states in that Hilbert space, see in particular 
Section \ref{sec:contributingStates}.
\newline
To bring the problem of solving the constraints of U(1)-BF theory to the computer, we have to make some drastic truncations of the infinitely many kinematical degrees of freedom. We will base the investigation on a fixed graph $\gamma$. However, the number of computational degrees of freedom in such a construction would be equal the number of charge network states on $\gamma$. With allowed charges $m \in \Z$, one clearly still has infinitely many degrees of freedom. Working with infinite computational degrees of freedom, however, often comes at a computational cost if at all feasible. Thus, in order to facilitate the implementation, one can consider deforming the algebra of the representation labels of the holonomies such that one is restricted to labels which fall within a certain admissible set of charges $M:= [-\jmax , \dots , \jmax] \subseteq \Z$. This can be done by requiring that given any two charge numbers $m, n \in \Z$, then
\begin{equation}
\label{eq:pbc}
    m \oplus n := (m + n + \jmax) \mod (2\jmax + 1) - \jmax.
\end{equation}
Clearly, this gives us a total of $2\jmax + 1$ allowed charges. We can then deform the algebra of holonomies such that for any given edge $e$
\begin{equation}
    h_e^m\, h_e^n= h_e^{m\, \oplus\, n}. 
\end{equation}
If we then generate charge network states from the AL-vacuum, it is ensured that their charges always lie in $[-\jmax , \dots , \jmax]$. This deformation creates a new problem, however. The modified holonomy operators are now not gauge covariant under general U(1) gauge transformations anymore. To see under which elements of U(1) they remain covariant, it is instructive to consider the edge case $\jmax \oplus 1 = -\jmax$. A viable $g \in \Uone$ would have to satisfy
\begin{equation}
    g^{\jmax} \circ g = g^{-\jmax} \quad \Leftrightarrow e^{i\varphi (\jmax + 1)} = e^{-i\varphi\jmax} \rightarrow \varphi_n = \frac{2n\pi}{2\jmax + 1},
\end{equation}
where $n \in \Z$, resulting in $2\jmax + 1$ group elements. This is nothing but the cyclic group $\Z_{2\jmax + 1}$, which can be identified with the quantum deformation $\Uqone$ of $\Uone$ at root of unity  \cite{2012arXiv1209.1135G}.
The set of charges $[-\jmax , \dots , \jmax]$ with the product \eqref{eq:pbc}
is then understood as the dual group $\Uqone^*$. 
\newline
In this case, we effectively impose the constraints of the continuum theory whereby we are still on its kinematical space, but now we have deformed the algebra of the holonomies to be restricted to a different (discrete) group. Unsurprisingly, this comes with subtle consequences. For example, the curvature and the Gau{\ss} constraint do not commute. Nevertheless, for large $|M|$, one approaches the behaviour of the continuum theory as will be seen later. We now refer to this model as a quantum $\Uqone$ BF-theory. We note that quantum deformation might also be an effective way to introduce a cutoff for a non-abelian gauge group. For the case of SU(2), see for example \cite{Dittrich:2018dvs} for some details.

\subsection{Solving physical systems numerically}
In a given quantum system, be it of a single- or many-particle nature, the wave-function $\Psi$ is the object which encompasses all the information needed to describe the quantum state of the system. In many-body physics problems, one encounters systems of interest which turn out to be very highly correlated. This leads the dimension of the system's Hilbert space to grow exponentially with its size which, in turn, means that one needs an exponential amount of information to fully determine $\Psi$. 
\newline
When looking to only describe the ground state however, one can get away with an exponentially less amount of information than that of a generic state. One either uses methods which rely on (i) efficient representation of the wave-function, or (ii) probabilistic, stochastic frameworks to efficiently sample different configurations of the system. For lower dimensions, owing to the limited entanglement entropy in the ground state, one uses (i) and one can employ techniques such as matrix product states (MPS) to describe the ground state \cite{White:1992zz,doi:10.1080/14789940801912366,Schollwoeck:2010uqf,Rommer:1997zz}. For higher dimensional systems, one resorts to using stochastic sampling such as Quantum Monte Carlo (QMC) \footnote{this applies for certain types of strongly correlated systems with positive-definite ground states} to efficiently obtain estimates of the physical quantities of interest \cite{Sandvik:1991axv,gubernatis_kawashima_werner_2016,1998PhLA..238..253P}. 
\newline
Both methods have shown rather remarkable capabilities and achievements. However, they do have their limitations. QMC for example is plagued by the sign problem ever since its conception \cite{Troyer:2004ge}, while the inefficiency of the compression approaches render them by definition unusable for higher dimensions. This leads to systems of interest which cannot be explored even numerically (e.g. finding the ground state of a strongly interacting fermionic system) \cite{2011RvMP...83..863P,doi:10.1142/1346,thouless1972quantum}.

\subsubsection{Neural network quantum states (NNQS)}
It is evident that the exponential complexity remains present in the task of representing a generic many-body wave-function. Consider for example a spin system of $M$ identical particles on a lattice. The Hilbert space would then be $\hilb = \bigotimes_{i}^{M}\hilb_{i}$ where $\hilb_i$ denotes the Hilbert space describing the $i$\textsuperscript{th} particle and spanned by the spin eigenstates $\ket{\sigma_i}$. A state in $\hilb$ can therefore be written as
\begin{equation}
\label{eq:wavefunction}
    \Psi(\underline{\sigma}) = \sum_{\sigma_1 \cdots \sigma_M} c_{\sigma_1 \cdots \sigma_M} \ket{\underline{\sigma}} \quad , \quad \ket{\underline{\sigma}} := \ket{\sigma_1}\otimes \cdots \otimes \ket{\sigma_M}
\end{equation}
To determine the state $\Psi(\underline{\sigma})$, one needs to determine the coefficients $c_{\sigma_1 \cdots \sigma_M}$. The idea then is to do this efficiently and using an (exponentially) smaller number of parameters than $\dim\hilb$. Therefore, one finds at hand a problem of dimensionality reduction and feature extraction. Such problems exist in other areas in sciences and engineering, and one very powerful tool that is used to address them is neural networks (NN) \cite{10.5555/523781}.
\newline
In a seminal work \cite{Carleo:2017nvk}, the landscape of numerical methods witnessed a transformative shift with the advent of neural network quantum states. NNQS is a representation for the wave-function which harnesses the power of artificial neural networks. Specifically, in \cite{Carleo:2017nvk}, it was first shown that a restricted Boltzmann machine (RBM) \cite{10.1007/978-3-642-33275-3_2} can be used as a variational ansatz for the Heisenberg model on a square lattice. The RBM is a type of neural network which employs a two layer (bipartite graph) structure which consists of a visible layer and a hidden layer. This simple yet unique architecture enables the RBM to capture the nature of the complex correlations which lie within a given quantum state. Additionally, RBMs are both generative and stochastic networks \cite{10.1007/978-3-642-33275-3_2} capable of learning probability distributions, which means one can train them using unsupervised learning, thus eliminating the typically tedious task in any machine learning problem of acquiring and appropriately labeling training data. The idea presented in \cite{Carleo:2017nvk} is broadly illustrated in the Figure \ref{fig:NNQS} below. 
\begin{figure}[h]
    \centering
    \includegraphics[scale=0.6]{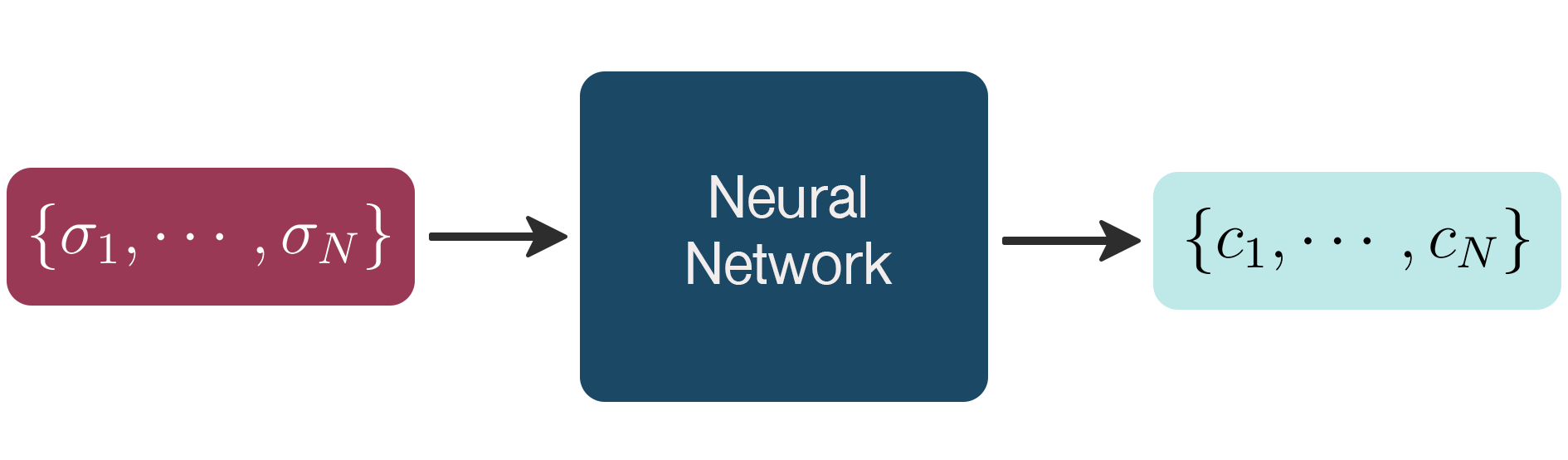}
    \caption{The general idea of the NNQS variational ansatz is shown in a simplified manner. From the left, one provides the set of basis states which characterise the physical model. Once provided to the neural network, it will produce for every basis state its amplitude  which contributes to the wave-function. These amplitudes are provided by the network such that the wave-function obtained minimises the expectation value of an operator, typically the Hamiltonian of the model.}
    \label{fig:NNQS}
\end{figure}
\newline
As shown in the figure, the problem is comprised of having batches of the basis states of the physical system as inputs, and the corresponding amplitudes as the output. This manner of interpreting the outputs of the network as the amplitudes allows us to write the wave-function in equation \eqref{eq:wavefunction} as
\begin{equation}
    \Psi(\underline{\sigma}) = \sum_{\sigma_1 \cdots \sigma_M} c_{\sigma_1 \cdots \sigma_M}^{\mathrm{NN}} \ket{\sigma_1}\otimes \cdots \otimes \ket{\sigma_M},
\end{equation}
where $c_{\sigma_1 \cdots \sigma_M}^{\mathrm{NN}}$ are the amplitudes determined by the RBM. One can then train the network such that it provides a specific set of amplitudes, namely the ones which minimise the expectation value of a given observable (see \cite{Carleo:2017nvk} and Appendices therein). Given that this observable would typically be the Hamiltonian of the system, one would then arrive at the ground state of the Hamiltonian, and its corresponding ground energy. 
\newline
This demonstrated not only the feasibility of employing RBMs, and more generally NNs, as a viable variational ansatz to solve quantum problems but also the broader potential of NNQS in solving several many-body quantum physics problems \cite{Carleo:2017nvk,PhysRevB.100.125124,PhysRevB.97.035116,Choo_2020}, performing well when compared to state of the art numerical methods  \cite{PhysRevB.100.125124}. The key being that using the NNQS one has the ability to efficiently parametrise highly entangled quantum states which are of importance in systems with strong correlations. 
\newline
Since the advent of NNQS, researchers have expanded upon this foundation of using an RBM architecture \cite{Nomura_2021,PhysRevLett.128.090501,Lin_2022}, exploring different architectures ranging from simple feed-forward neural networks to convolutional neural networks \cite{PhysRevB.100.125124} and more. The application of NNQS saw itself being realised not only in several fields of physics but also beyond the scope of physics such as in quantum chemistry \cite{Choo_2020,Wu:2023gic,Wu:2023kwx}. 
\newline
While research continues to refine and expand upon this ansatz, it has remained foreign to the field of quantum gravity. In this work, we take the first step into applying NNQS in a quantum gravity setting by considering the simple $\Uone$ BF-theory quantised using LQG methods as described above hence demonstrating the applicability of NNQS to this domain. The paper is divided into two parts. The first part concerns the implementation of the physical model on the computational basis and the second discusses the results obtained. The exact structure of the paper is as follows:
\begin{enumerate}
    \item In Section \ref{sec:implementingThePhysicalSystem}, the implementation of the quantum $\Uqone$ BF-theory model on a computational basis is presented whereby the exact action and definition of the operators as well as the terminology related to the computational basis are defined.
    \item A discussion over the architecture of the neural network used in this work is presented in Section \ref{sec:networkArchitecture}.
    \item The last part of the paper reports on the results obtained in which we discuss the following:
    \begin{enumerate}
        \item[$\bullet$] In Section \ref{sec:groundState}, the nature of the solutions to the quantum Hamilton constraint at different charge cutoffs is discussed.
        \item[$\bullet$] Section \ref{sec:quantumFluctuations} concerns the quantum fluctuations of the observables of the theory.
        \item[$\bullet$] Next, a discussion on the type of basis states which contribute to the solution is presented in Section \ref{sec:contributingStates}.
        \item[$\bullet$] Lastly, a discussion on the entropy of the solution and the entanglement entropy is done in Section \ref{sec:entropy}.
    \end{enumerate}
    \item The appendices provide some details on quantum states with real coefficients. Further, some technical details related to the simulation process are presented whereby a brief discussion on variational Monte-Carlo (VMC) methods, Markov chain Monte-Carlo (MCMC) and some statistical tools necessary in both of these processes are discussed.
\end{enumerate}

\section{The physical and computational model}

In this section we address two points, namely (i) the exact implementation of the above physical model in a computational basis and (ii) the architecture of the neural network used in this work for the NNQS variational ansatz. A brief discussion about the technical details concerning the tools employed in the process concerning the neural network of the work are presented at the end of this section.

\subsection{Implementing the quantum $\Uqone$ BF-theory}
\label{sec:implementingThePhysicalSystem}
We begin by setting the terminology. Our quantum $\Uqone$ BF-theory model is defined on oriented graphs the vertices of which are labeled by integers $v$ and the edges by tuples $e_{if} := (v_i, v_f)$ where $v_i, v_f$ denote the vertices it emanates from and is incident at respectively.
\newline
We view now our oriented graph as an undirected graph, thus momentarily ignoring orientation. In doing so, one can identify the cycle space of the graph and moreover, the cycle basis of the cycle space \cite{LIEBCHEN2007337}. The cycle basis can include cycles of different lengths. In our work, we will denote by the \textit{minimal loops} of a graph the set of cycles in the cycle basis of the graph which have the smallest length. By \textit{$N$-L graph} we mean a graph which has $N$ minimal loops. We choose $L_K(\gamma)$ to denote the set of minimal loops of a graph $\gamma$ where $K$ denotes their length. We will often drop the subscript $K$ and use $L(\gamma)$ only. In our work, we will consider graphs which are \textit{non-trivial} whereby $|L(\gamma)| \geq 2$\footnote{For $N=1$, the Gau{\ss} constraint would be identically satisfied}. Lastly, every minimal loop will have an \textit{incrementing orientation}. That is, a minimal loop $\alpha \in L_K(\gamma)$ passes through some vertices $v_0 \rightarrow v_1 \rightarrow \dots \rightarrow v_K$ such that $v_0 < v_1 < \dots < v_K$. Therefore, the orientation of the minimal loops need not coincide with the orientation of $\gamma$. 
\newline
As the edges carry the charges, our physical system in the computational basis will be defined on the \textit{dual graph} $\Tilde{\gamma}$ of the given graph $\gamma$. This is required as all the computational tools used in this work require the Hilbert space to be defined on the vertices of the graph. In our work, to construct the dual $\Tilde{\gamma}$ of a graph $\gamma$ we mean that one constructs a new graph $\Tilde{\gamma}$ from $\gamma$ such that now every edge $e \in E(\gamma)$ is represented by a vertex $\Tilde{v} \in V(\Tilde{\gamma})$ and every vertex $v \in V(\gamma)$ is represented by an edge $\Tilde{e} \in E(\Tilde{\gamma})$. The dual graphs are taken to be non-oriented. Rather, they are used only in the computational basis to carry the action of the operators, respecting the orientation, as they would be applied on the graph itself. Thus, while the physics is defined on $\gamma$, the computational model is implemented on $\Tilde{\gamma}$. Note that we will use the terminology \quotes{dual vertex/edge} to refer to $\Tilde{v}$ and $\Tilde{e}$ respectively.
\newline
We denote by $\jmax \in \Z$ the maximal value of charge allowed in the model. The edges then can carry any charge $m \in \Z$ from the set of $2\jmax + 1$ charges denoted by $M := [-\jmax, \dots , \jmax] \subseteq \Z$. In this picture, an analogy one can draw is that every dual vertex can be understood as a particle whose state is written in the angular momentum basis with allowed angular momentum numbers $m \in M$. Hence, the system at hand is equivalent to an identical $N$-body system where the Hilbert space on every dual vertex $\hilb_{\Tilde{v}}$ is spanned by the basis states labeled by $m \in M$. As a result, the Hilbert space over the entire graph is nothing but the $N$-fold tensor product of $\hilb_{\Tilde{v}}$. Such a Hilbert space then has dimensions of $\dim\hilb_{\Tilde{\gamma}} = (2\jmax + 1)^N$. 
\newline
Now we are ready to look at the implementation of the operators in the computational basis. It was seen in prior sections that the constraints of theory are composed of a curvature term $\hat{F}$ and a Gau{\ss} term $\hat{G}$. We begin with the curvature term which enforces flatness of the holonomies of loops in our graph. We now define a \textit{minimal loop holonomy operator} and denote it by $\hat{h}_\alpha$ such that 
\begin{equation}
    \hat{h}_{\alpha} = :\prod_{v \in V(\alpha)} \overset{\leftrightarrow}{\hat{h}}_v:,
\end{equation}
where by $:\hat{O}(v_1) \hat{O}(v_2):$ we denote the ordered product of operators such that $:\hat{O}(v_1) \hat{O}(v_2): = \hat{O}(v_1) \hat{O}(v_2)$ for $v_1 < v_2$. The minimal loops can have a different orientation with respect to the graph. Hence, the $\leftrightarrow$ indicates that the holonomy operators $\hat{h}_v$ are either as they are or their adjoint depending on the orientation of the loop with respect to the underlying graph. Since holonomy operators in the $\Uqone$ gauge act as multiplication operators changing the charge numbers, then $\hat{h}_\alpha$ are graph preserving, merely just changing the charges on the edges included in $\alpha$ by acting as raising or lowering operators with respect to the charge numbers. To keep the charge of edge within the allowed set of charges, the periodic boundary condition \eqref{eq:pbc} is imposed. The basic holonomy operators in the computational basis then take the form
\begin{equation}
\hat{h}_v = \begin{pmatrix}
0 & 1 & 0 & \ldots & 0 \\
0 & 0 & 1 & \ddots & \vdots \\
\vdots & \vdots & \ddots & \ddots & 0 \\
0 & 0 & \ldots & 0 & 1 \\
1 & 0 & \ldots & 0 & 0 \\
\end{pmatrix} 
\quad , \quad
\hat{h}_v^\dagger = \begin{pmatrix}
0 & 0 & \ldots & 0 & 1 \\
1 & 0 & \ldots & 0 & 0 \\
0 & 1 & \ddots & \vdots & \vdots \\
\vdots & \ddots & \ddots & 0 & 0 \\
0 & \ldots & 0 & 1 & 0 \\
\end{pmatrix},
\end{equation}
where the 1 in the bottom leftmost or upper rightmost entry enforce the periodic boundary condition. For the case of our gauge group, these basic holonomy operators commute and hence, the ordering of operator products in the minimal loop holonomy operator can be neglected. 
\newline
The curvature constraint can be faithfully taken to enforce flatness over the minimal loops of the graph, as every loop in the graph can be written in terms of minimal loops. In this work, we will construct a master constraint, and hence we consider squares of the constraints and we do not consider any smearing. The curvature constraint then takes the form 
\begin{equation}
    \hat{F} = \sum_{\alpha \in L(\gamma)} \left(\hat{h}_{\alpha} - \one\right)\left(\hat{h}_{\alpha}^\dagger - \one\right) = \sum_{\alpha \in L(\gamma)} \left(-\hat{h}_{\alpha} - \hat{h}_{\alpha}^\dagger + 2\one\right).
\end{equation}
The remaining Gau{\ss} constraint, in this context, translates to charge conservation at every vertex which is expressed as the positive operator
\begin{equation}
    \hat{G} = \sum_{v \in V(\gamma)} \left(\sum_{e \in E_i(v)} \hat{N}_{e} - \sum_{e' \in E_o(v)} \hat{N}_{e'}\right)^2.
\end{equation}
Here, by $E_i(v)$ and $E_o(v)$ we mean edges attached to $v$ which are either incident at or outgoing from $v$ respectively. By $\hat{N}_e$ we denote an operator which just returns the charge at the edge it acts on. In the computational basis, this is a $\dim\hilb_{\Tilde{\gamma}} \times \dim\hilb_{\Tilde{\gamma}}$ diagonal matrix with $\mathrm{diag}(\jmax , \dots , -\jmax)$. To this end, we now have all the tools necessary to construct this quantum $\Uqone$ BF-theory and its constraints on the computational basis. The constraint which we solve is a master constraint which is denoted by 
\begin{equation}
    \hat{C} = \hat{F} + \hat{G}.
\end{equation}
Given the formulation, we see that this constraint is not graph changing. Consequently, one can obtain a matrix of size $\dim\hilb_{\Tilde{\gamma}}\times\dim\hilb_{\Tilde{\gamma}}$ for the constraint $\hat{C}$. If one understands this as the Hamiltonian of the quantum model, then from then onwards one can both use the NNQS ansatz and perform exact diagonalisation using standard numerical techniques (e.g. Lanczos algorithm \cite{Lanczos:1950zz}) to obtain the minimum eigenvalue and the corresponding eigenvector. These eigenvectors, then, would be the solutions of the constraint $\hat{C}$.

\subsection{The network architecture}
\label{sec:networkArchitecture}
In this section, we provide a general outline for the network architecture used in this work. For a detailed discussion, see Appendix \ref{app:D}.It was observed in our work that in such $\Uqone$ BF-theory models, the RBM is not the optimal neural network architecture. As a result, our wave-function is instead parametrised as a variant of a specific type of feed-forward neural network (FFN), referred to in the literature as a convolutional neural network (CNN) \cite{OShea2015AnIT,Yamashita2018}. A feed-forward network is a network composed of a series of transformations applied to the input of the network. These transformations are referred to as \textit{layers}. Thus, a $k$ layer feed-forward network will take an input $I_{(0)}$ and use $k$ layers to map it to the output $I_{(k)}$. Generally, a layer implements an affine transformation such as
\begin{equation}
\label{eq:ffn}
    I_{i}^{(n)} = f\left(b_{i} + \sum_{j} w_{ij}I_{j}^{(n - 1)}\right).
\end{equation}
Here, $w_{ij}$ are called \textit{weights} and $b_i$ are called \textit{biases} (see Appendix \ref{app:A} on how the network parameters (weights and biases) are updated during the training process to achieve convergence). This operation is conducted element-wise. Lastly, the function $f$ is simply some non-linear function which introduces non-linearity to the network. 
\newline
A CNN has spatially local structures, the \textit{convolution} layers. Just as in equation \eqref{eq:ffn}, the CNN applies a similar type of transformation. However, some constraints are imposed on the weights and the biases \cite{OShea2015AnIT,Yamashita2018}. The transformation is separated into several \textit{channels} which are characterised respectively by a matrix called the \textit{kernel} (or \textit{filter}). The convolution layers scan the input with the kernels and extract local patterns and features. Each kernel is applied across different regions of the input, allowing the network to detect spatial hierarchies. By choosing different sized kernels, one can extract fine or coarse features of the data. By sharing the parameters across different spatial locations, the network becomes more efficient in learning and recognising patterns, thus reducing the overall size of the network compared to an FFN. Lastly, CNNs often will include \textit{pooling} layers, which down-sample the dimensions of the data, thus reducing the computational complexity while retaining and focusing on the most important features.
\newline
Thus, the CNN employs a transformation more complicated compared to the one shown in equation \eqref{eq:ffn} but nevertheless carries the same principle. Furthermore, it has been shown in the literature the CNN type networks can be more effective than RBMs when employed to solve systems where the entanglement entropy follows a volume-law \cite{PhysRevLett.122.065301}. Moreover, recent work has shown that CNNs perform well with frustrated physical models \cite{PhysRevB.100.125124}. This can also be seen in our case, as the operators in our constraint do not commute after imposing the charge cutoff, thus introducing in a sense frustration. The takeaway point here is that some problems lend themselves by nature to be solved by a CNN rather than an RBM. Motivated by all the above, we construct a network which includes convolutions the general outline of which is shown in Figure \ref{fig:ourArchitecture}.
\begin{figure}[h]
    \centering
    \includegraphics[width=\textwidth]{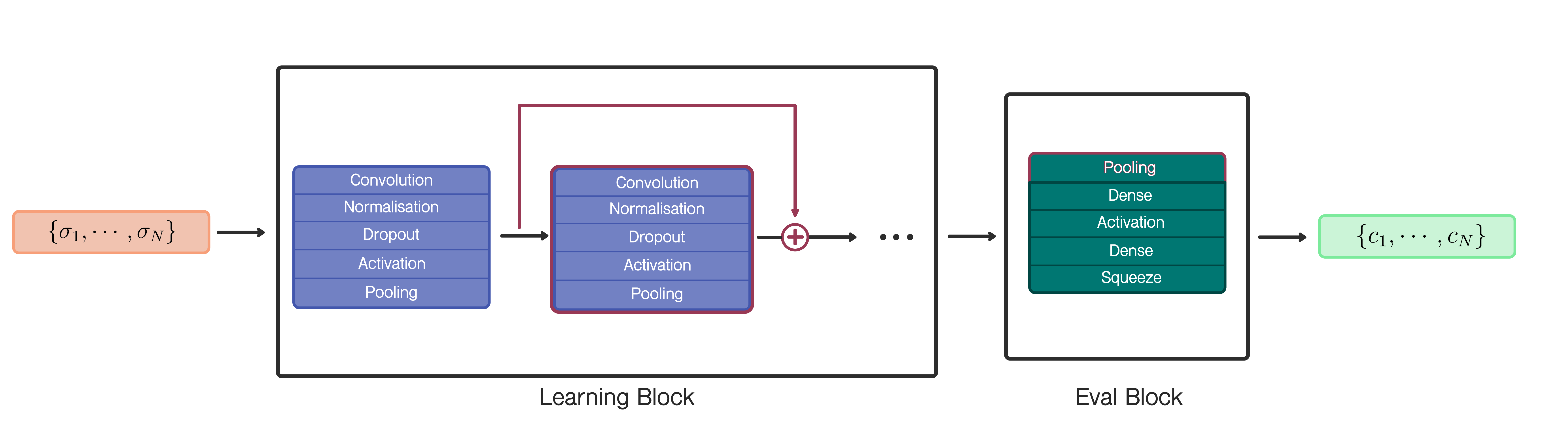}
    \caption{The architecture of the neural network used in this work. From the left, the network accepts batches of MCMC sampled configurations. The \textit{Learning Block} includes several convolution blocks, shown in purple, each composed of different layers. Skip connections are shown in red. Lastly, the \textit{Eval Block} is a simple FFN the output of which, shown at the utmost right, is comprised of the amplitudes of each basis state in the given batch of configurations.}
    \label{fig:ourArchitecture}
\end{figure}
\newline
As can be seen in Figure \ref{fig:ourArchitecture}, the network is composed of two major blocks: the \textit{learning block} and the \textit{(eval)uation block}. The first is tasked to learn, that is to perform feature extraction on the given inputs. It is composed of several convolutional sub-blocks. These sub-blocks include convolution layers but also normalisation\footnote{A layer normalising the input to the activation layer to prevent instabilities and fluctuations}, dropout\footnote{A layer where some weights are randomly \quotes{switched off} to prevent the network from biasing certain weights}, activation and pooling layers. The number of these convolution sub-blocks varies with respect to $\dim\hilb_{\Tilde{\gamma}}$. Furthermore, for large $\jmax$, then skip connections\footnote{A type of transformation which connects different previous layers to each other, these can be \textit{long} or \textit{short} skip connections.} resembling what would be implemented in a residual network \cite{7780459} are implemented, as shown in red in the figure. The layers of the convolution sub-blocks then change to adhere to that of a residual network. Lastly, the eval block is a simple FFN, the output of which is then understood as the amplitudes of the wave-function. The non-linearity is introduced by using the hard sigmoid linear unit (SiLU) activation in most blocks, and the sigmoid or rectified linear unit (ReLU) in others, such as the FFN block. Thus, our non-linearity can take the forms
\begin{equation}
    f(x) =
    \begin{cases}
        \text{(i) Sigmoid} & : \sigma(x) = 1/(1 + e^{-x}) \\
        \text{(ii) ReLU} & : \max(0, x), \\
        \text{(iii) Hard SiLU} & 
        : \begin{cases}
            x \cdot \sigma(x), & \text{if } x > \text{threshold} \\
            0, & \text{if } x \leq \text{threshold}
        \end{cases}
    \end{cases}
\end{equation}
We note that certain neural networks, which are inherently invariant or equivariant under certain groups, have been used \cite{Luo:2022jzl,PhysRevLett.121.167204,PhysRevResearch.5.013216}. However, we opt not to implement such networks in our work. We will show that despite that, the network is able to respect the gauge, and effectively narrowing down the search for only the gauge invariant states. 
\newline
Lastly, we saw that all the operators in the computational basis are real valued. The wave-function in general has complex valued amplitudes. This would mean that the network used in the NNQS ansatz should have complex valued parameters to produce complex valued amplitudes. Nevertheless, in the work presented we will work with real valued networks, which will nevertheless lead us to the solutions of the constraints (see Appendix \ref{app:C}). Note, however, that this choice in the computational basis is strictly for the sake of computational costs. In principle, this is implementable but comes at the cost of requiring more computing resources.

\subsection{A note on the technical details}
The industry standard for deep learning problems is to automate hyperparameter tuning. This ensures that the parameters needed for the network to run and produce optimal results are not chosen arbitrarily and moreover are the most optimal ones. In our work, these hyperparameters include the learning rate as well as other more technical related parameters. This has also been done in our work, albeit the technical details are left out. More importantly, to use the NNQS ansatz, a variational Monte-Carlo (VMC) \cite{gubernatis_kawashima_werner_2016} process is done, where Markov Chain Monte-Carlo (MCMC) \cite{Dellaportas2003} simulations are conducted (see Appendix \ref{app:A} for an outline and \cite{Carleo:2017nvk} for more details). In MCMC, one can have different samplers (e.g. Metropolis-Hastings type samplers \cite{Chib1995}) and the MCMC process has its own hyperparameters (e.g. number of chains for the MCMC). To ensure that the MCMC has converged, the chains have reached an equilibrium and that one indeed can reliably trust the results, the $\hat{R}$ (Gelman-Rubin) statistic is used \cite{doi:10.1080/10618600.1998.10474787}. For $\hat{R}$ values larger than 1.1 (or even 1.01) \cite{Gelman2013BayesianDA,Vehtari_2021}, one in principle should discard the samples and consequently the results. In our case, we suggest to use a higher $\hat{R}$ threshold for such models than the ones commonly used in the literature if one uses Metropolis type sampler. The reason for that, as well as the empirical results which will support the validity of the obtained data, is discussed in Appendix \ref{app:B}. All the results presented in this work are obtained using an exact sampler using exact inference methods, and thus are \textit{not} subject to this argument.

\section{Results}
\label{sec:results}

In this section, we discuss the results obtained by solving our physical model as described above using the NNQS ansatz. Specifically, we will present and discuss: (i) the nature of the solution of $\hat{C}$ for different charge cutoffs, (ii) the quantum fluctuations of the minimal loop holonomies, (iii) the contributing basis states in the solution for different charge cutoffs and lastly, (iv) the entropy and the second R\'{e}nyi entanglement entropy of the solution.

\subsection{Computing $\min\expect{\hat{C}}$ and solutions of $\hat{C}$}
\label{sec:groundState}
The starting point is to fix the graph on which we define our model on. We will work with the smallest, non-trivial graph (contains at least two minimal loops). The results obtained here are not graph dependent as 3-L as well as 4-L graphs have been also studied. The choice of presenting a 2-L graph model is solely based on choosing the setup with the least computational cost, allowing for the most amount of numerical results. The graph in use is shown in Figure \ref{fig:graph}.
\begin{figure}[h]
    \centering
    \begin{subfigure}{.4\textwidth}
        \centering
        \includegraphics[width=1\linewidth]{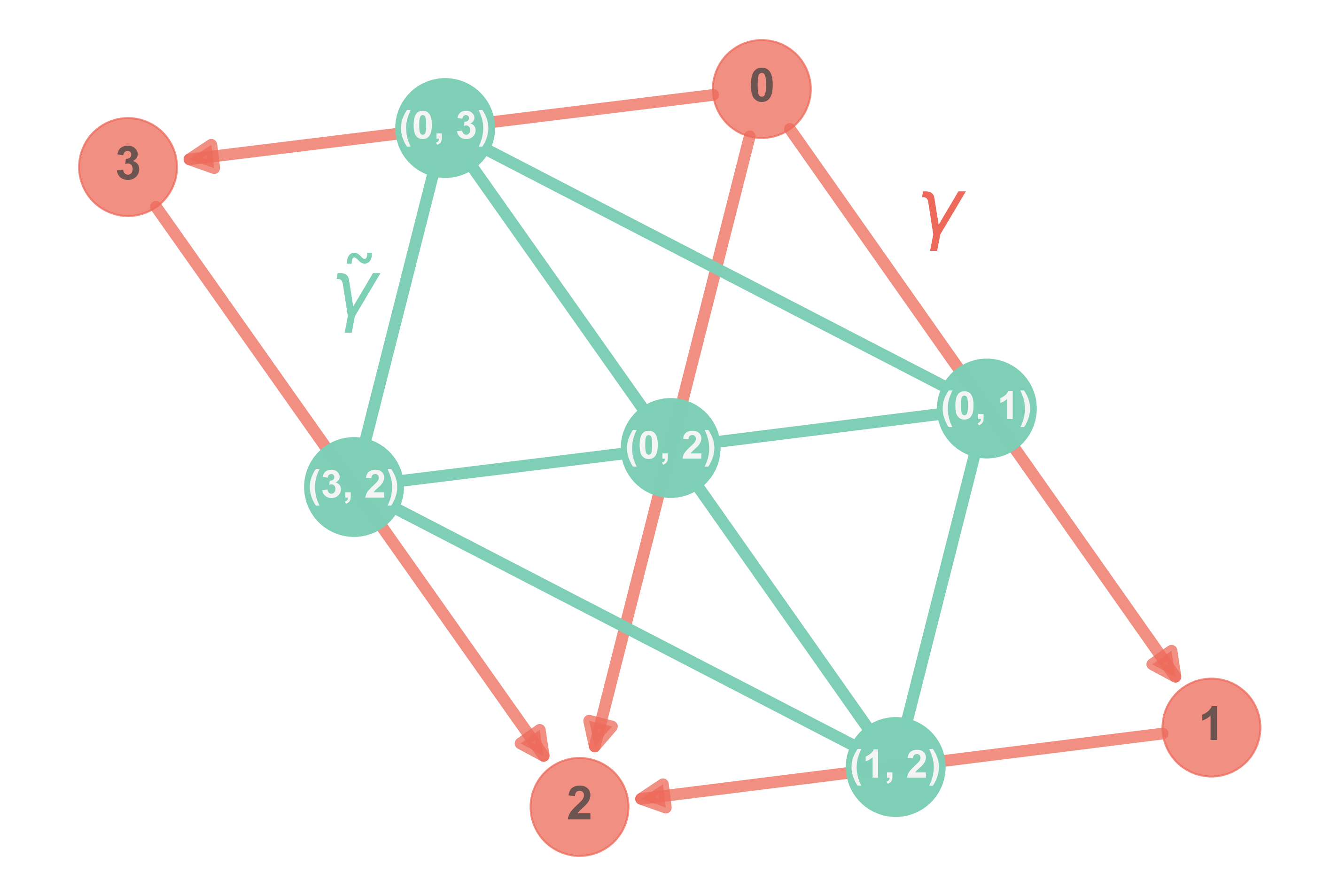}
        \caption{}
        \label{fig:graph_sub1}
    \end{subfigure}%
    \begin{subfigure}{.4\textwidth}
        \centering
        \includegraphics[width=1\linewidth]{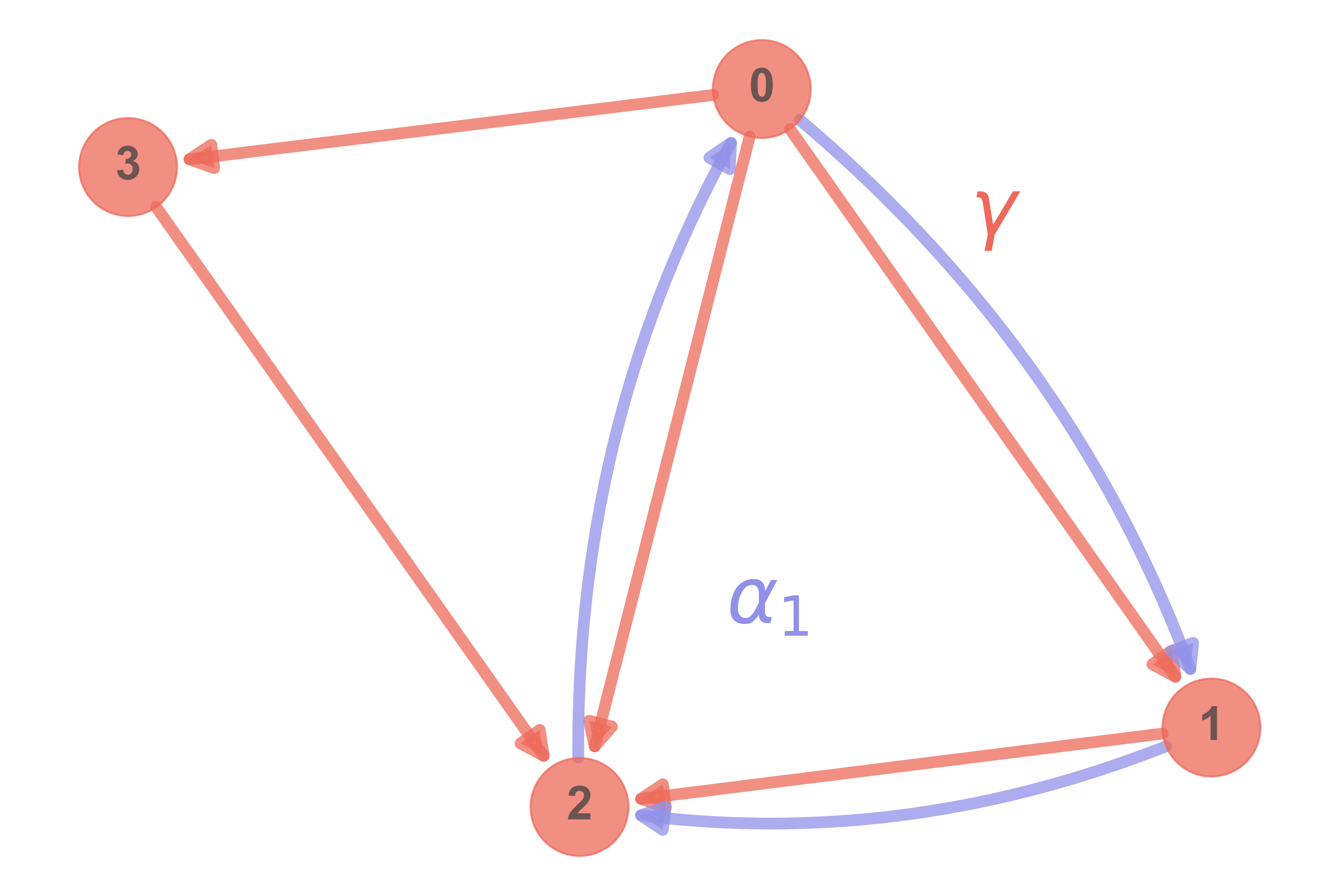}
        \caption{}
        \label{fig:graph_sub2}
    \end{subfigure}
    \caption{On the left, an illustration of the chosen oriented graph $\gamma$ shown in red and its dual graph $\Tilde{\gamma}$ shown in green. On the right one of the two minimal loops $\alpha_1$ of $\gamma$ is shown in purple. The edges of $\alpha_1$ are curved purely for illustrative purposes.}
    \label{fig:graph}
\end{figure}
\newline
In Figure \ref{fig:graph_sub1}, the 2-L graph $\gamma$ shown in red denotes the graph on which the analytical model is defined on. One of the minimal loops is shown on the right in Figure \ref{fig:graph_sub2}. The edges of $\alpha_1$ are curved purely for illustrative purposes. The dual graph $\Tilde{\gamma}$ denotes the graph on which the computational basis is defined on.
\newline
Imposing now a charge cutoff on the model, this gives us for an example of $\jmax = 2$, a Hilbert space with dimension $\dim\hilb_{\Tilde{\gamma}} = 3125$. Since our networks are not inherently gauge invariant, that means that we explore the entire Hilbert space, not only the gauge invariant subspace, and thus we need to determine the amplitudes of 3125 basis states which give a minimum value for $\expect{\hat{C}}$ to fully be able to write down the expression of the solution $\Psi$. With the exception of the $\jmax = 1$ case, the number of parameters in our network is always less than $\dim\hilb_{\Tilde{\gamma}}$. For example, for $\jmax = 6$, the network has only $5812$ parameters to be optimised, constituting only a small fraction (1.57\%) of the dimensions of the Hilbert space (371293). Now, using the NNQS ansatz, it is possible to find the solution to $\hat{C}$ using a significantly lower amount of information than $\dim\hilb_{\Tilde{\gamma}}$.
\newline
For such a small graph, the system is also exactly diagonalisable even for relatively high charge cutoff. Here, exact diagonalisation (ED) is simply finding the lowest eigenvalue for the matrix which represents our constraint using standard numerical techniques. The method used in this work is the Lanczos algorithm. As shown in Figure \ref{fig:spin5_results}, the network quickly converges to the exact diagonalisation result shown in red, and one does not need many iterations to reach a good accuracy. This $\jmax = 5$ simulation shown in the figure concluded with an accuracy in $\min \expect{\hat{C}}$ of approximately 92.819\%. 
\begin{figure}[h]
    \centering
    \includegraphics[scale=0.4]{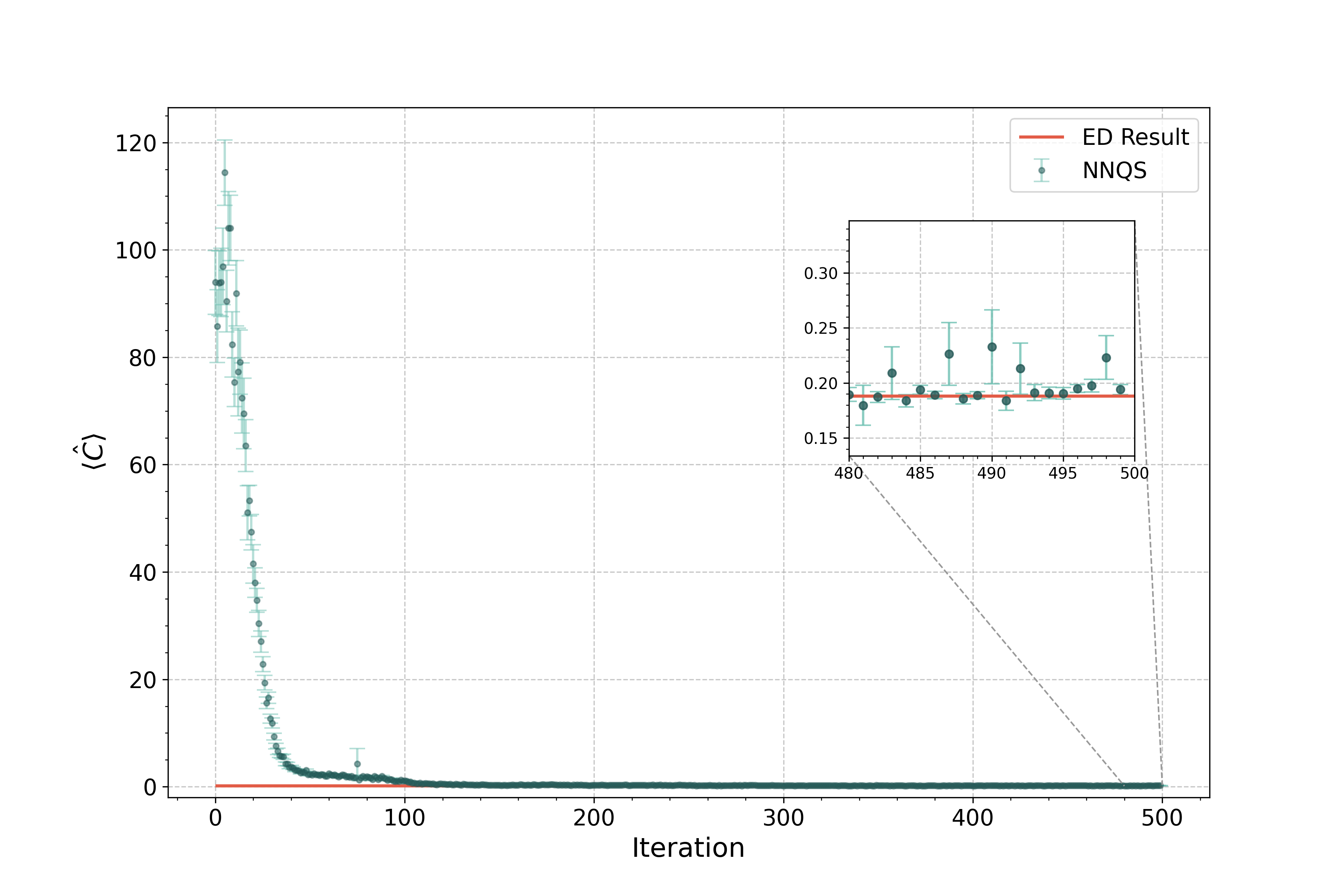}
    \caption{A $\jmax = 5$ simulation is shown where the NNQS ansatz is used to find the solution for the master constraint of the $\Uqone$ BF-theory posed on the graph described in Figure \ref{fig:graph}. The exact diagonalisation result for $\min\expect{\hat{C}}$ is shown in red and the NNQS result is shown in green for every iteration in the simulation. The accuracy in this case is 92.819\%.}
    \label{fig:spin5_results}
\end{figure}
\newline
In Table \ref{tab:groundstate_results}, we show the results different simulations for different charge cutoffs ranging from $\jmax = 1$ to $8$. The ED results are compared with the results obtained from the NN. It is evident that the network is capable to obtain and maintain high accuracy relative to the ED result, indicating that the architecture employed can be used for solving such models. This becomes especially useful once the considered models become too large to be exactly diagonalised as the high accuracy establishes a degree of confidence in the results obtained even if one cannot compare it with ED results.
\begin{table}[h]

    \centering
    \begin{tabular}{c|cccc}
    \rowcolor{lightergreen!20}
        $\jmax$ & $\min\expect{\hat{C}}_{\mathrm{ED}}$ & $\min\expect{\hat{C}}_{\mathrm{NN}}$ & Accuracy (\%) & $\langle\Psi^{(\mathrm{NN})} | \Psi^{(\mathrm{ED})} \rangle^2$ \\
        \hline
        1 & 0.835968 & 0.99866 ± 0.00028 & 80.538 & 0.9895 \\
        2 & 0.601165 & 0.625 ± 0.0013 & 96.034 & 0.9979 \\
        3 & 0.389553 & 0.3942 ± 0.0045 & 98.818 & 0.996 \\
        4 & 0.263623 & 0.2648 ± 0.0013 & 99.539 & 0.9906 \\
        5 & 0.187973 & 0.1882 ± 0.0017 & 99.853 & 0.989 \\
        6 & 0.140084 & 0.1412 ± 0.0035 & 99.189 & 0.9834 \\
        7 & 0.108159 & 0.1133 ± 0.0066 & 95.218 & 0.983 \\
        8 & 0.085918 & 0.0833 ± 0.0079 & 96.931 & 0.9598 \\
    \end{tabular}
    \caption{The value of $\min\expect{\hat{C}}$ at different charge cutoff values are shown. The exact diagonalisation results (ED) are compared to the results from the neural network (NN). Note that values are truncated to 6 decimal values at most.}
    \label{tab:groundstate_results}
\end{table}
\newline
Further seen in Table \ref{tab:groundstate_results} is the effect of having discrete group $\Uqone$ instead of the continuum $\Uone$ theory whereby now one obtains $\min\expect{\hat{C}} > 0$. This is despite the fact that it was observed, using the NN and exact diagonalisation, that $\min\expect{\hat{F}}$ and $\min\expect{\hat{G}}$ are independently is zero. Nevertheless, one sees that even for relatively low $\jmax$, one starts to recover the continuum theory as one sees that $\min\expect{\hat{C}} \rightarrow 0$ for $\jmax \rightarrow \infty$. 
\newline
We attribute the lower accuracy of the $\jmax = 1$ cutoff case to the fact that the network is probably overfitting due to the number of parameters being more than $\dim\hilb_{\Tilde{\gamma}}$. Furthermore, as shown on the two rightmost columns, the accuracy of the results obtained from the NN are measured in two ways: (i) the relative error compared to the value of $\min\expect{\hat{C}}$ as obtained from exact diagonalisation and (ii) the inner product of the solution with that of obtained from exact diagonalisation. We see that in fact, we not only have a good accuracy in terms of being close to the true $\min\expect{\hat{C}}$, but also that the solutions obtained via ED and the NN are nearly identical, as evident from their inner product, indicating that we do converge to the true solution as obtained from exact diagonalisation. Lastly, we note that the error bars in this work are computed as shown in Appendix \ref{app:A}.

\subsection{Quantum fluctuations of minimal loop holonomies}
\label{sec:quantumFluctuations}
Once the solution has been obtained, one can compute expectation values of observables. In all cases of different $\jmax$, it was seen that $\expect{\hat{G}}$ is almost 0. This shows that gauge invariance is almost perfectly imposed. In fact, it seems to be the case that the Gau{\ss} part $\hat{G}$ of the constraint is often more enforced than the curvature counterpart. This eliminates the need to implement gauge invariant neural networks, as the current network has no trouble narrowing down the search to gauge invariant states. 
\newline
In our simple model, one observable of interest that exist is the minimal loop holonomy $\hat{h}_{\alpha_k}$. Additionally, the curvature part of the master constraint effectively imposes that these minimal loop holonomies are as close to $\one$ as possible. Once the solution is obtained, in principle one expects that computing $\expect{\hat{h}_{\alpha_k}}$ should yield a value as close to 1 as possible if the constraint is well satisfied and the solution is accurate. Moreover, we can also look at quantum fluctuations in $\hat{h}_{\alpha_k}$, which we define as
\begin{equation}
    \Delta \hat{h}_{\alpha_k} := \expect{(\hat{h}_{\alpha_k} + \hat{h}_{\alpha_k}^\dagger)^2} - \expect{\hat{h}_{\alpha_k} + \hat{h}_{\alpha_k}^\dagger}^2.
\label{eqn:q_fluctuations}
\end{equation}
In our graph, $L(\gamma) = 2$ and thus we have two observables to compute for. The following table shows the results for different charge cutoffs, each obtained in the solution of the corresponding simulation shown in the Table \ref{tab:groundstate_results}.
\begin{table}[h]
    \centering
    {\renewcommand{\arraystretch}{1.5}
    \begin{tabular}{c|cccccccc}
        \rowcolor{lightergreen!20}
        \multirow{2}{*}{\cellcolor{lightergreen!20} Observable} & \multicolumn{8}{c}{\cellcolor{lightergreen!20} Charge Cutoff ($\jmax$)} \\
        
        \cline{2-9} 
        \rowcolor{lightergreen!20}
        \multirow{-2}{*}{\cellcolor{lightergreen!20} Observable}
         & 1 & 2 & 3 & 4 & 5 & 6 & 7 & 8 \\
        \hline

        $\expect{\hat{h}_{\alpha_1}}$ & 0.769 & 0.834 & 0.878 & 0.893 & 0.947 & 0.981 & 0.969 & 0.968 \\
        $\Delta \hat{h}_{\alpha_1}$ & 1.281 & 0.551 & 0.179 & 0.106 & 0.046 & 0.004 & - & - \\

        \hline

        \multicolumn{9}{c}{} \\[-4.6ex]
        \hline

        $\expect{\hat{h}_{\alpha_2}}$ & 0.746 & 0.819 & 0.92 & 0.931 & 0.95 & 0.955 & 0.986 & 0.981 \\
        $\Delta \hat{h}_{\alpha_2}$ & 1.215 & 0.581 & 0.284 & 0.089 & 0.071 & 0.012 & - & - \\
        
        \hline
    \end{tabular}
    }
    \caption{Results of the expectation value and quantum fluctuations of $\hat{h}_{\alpha_k}, k = 1, \dots , |L(\gamma)|$ in the solution of $\hat{C}$ computed by the neural network. The results are truncated to the third decimal point. The fluctuations for the $\jmax = 7, 8$ case are omitted due to numerical arithmetic precision errors.}
    \label{tab:min_loop_holonomy_results}
\end{table}
\newline
As shown in Table \ref{tab:min_loop_holonomy_results}, the expectation values of the minimal loop holonomies in the state obtained by the neural network after the minimisation process tend to get closer to 1 the higher the $\jmax$. We note that the error bars in the values shown in the Table above have been omitted only for brevity and in fact, the error values were less than 5\% of the result in all cases. This happens even at a low $\jmax$, and hence indicating that one can already see that one approaches the behaviour of the continuum theory relatively quickly. This also further indicates that the flatness enforced by the curvature constraint becomes more so enforced for higher $\jmax$ as the frustration between the $\hat{F}$ and $\hat{G}$ parts of the constraint get alleviated. Looking at the quantum fluctuations, one sees that they tend to get closer to zero as $\jmax$ gets larger. We note that for high $\jmax$, the two terms in equation (\ref{eqn:q_fluctuations}) have very close numerical values and unavoidable numerical arithmetic precision errors can lead the result to be negative, albeit analytically this should not be the case. For this reason we opt not to include the results for $\jmax = 7, 8$ in Table \ref{tab:min_loop_holonomy_results}. Nevertheless, one can infer with some degree of confidence that the quantum fluctuations for higher charge cutoffs die out steadily as for accurate simulations one does not encounter values for the quantum fluctuations at a higher $\jmax$ which are larger than that of a lower $\jmax$, especially if the accuracy achieved during the simulations is high.

\subsection{Contributing basis states}
\label{sec:contributingStates}
In the previous sections, the state satisfying the constraint $\hat{C}$ was obtained for different $\jmax$ values using NNQS. In this process, we obtained the amplitude for every basis in our system. Knowing that, for small systems (small enough to fit into memory during runtime), one can examine these amplitudes. In this section, we will look at what basis states contribute to the solution. 
\newline
First, we will define by a \textit{strongly} contributing basis state a basis state which has an amplitude which exceeds a chosen cutoff which is taken as $\mu + \epsilon\sigma$, where $\mu$ is the mean of all amplitudes, $\epsilon$ is a specified threshold often taken to be $10^{-3}$ and $\sigma$ is the standard deviation of the amplitudes. Visualising the amplitudes of the solution obtained via the network in the $\jmax = 2$ cutoff, we see that there are in total 19 strongly contributing basis states as shown in Figure \ref{fig:spin2_contributingStates}.
\begin{figure}[h]
    \centering
    \includegraphics[scale=0.45]{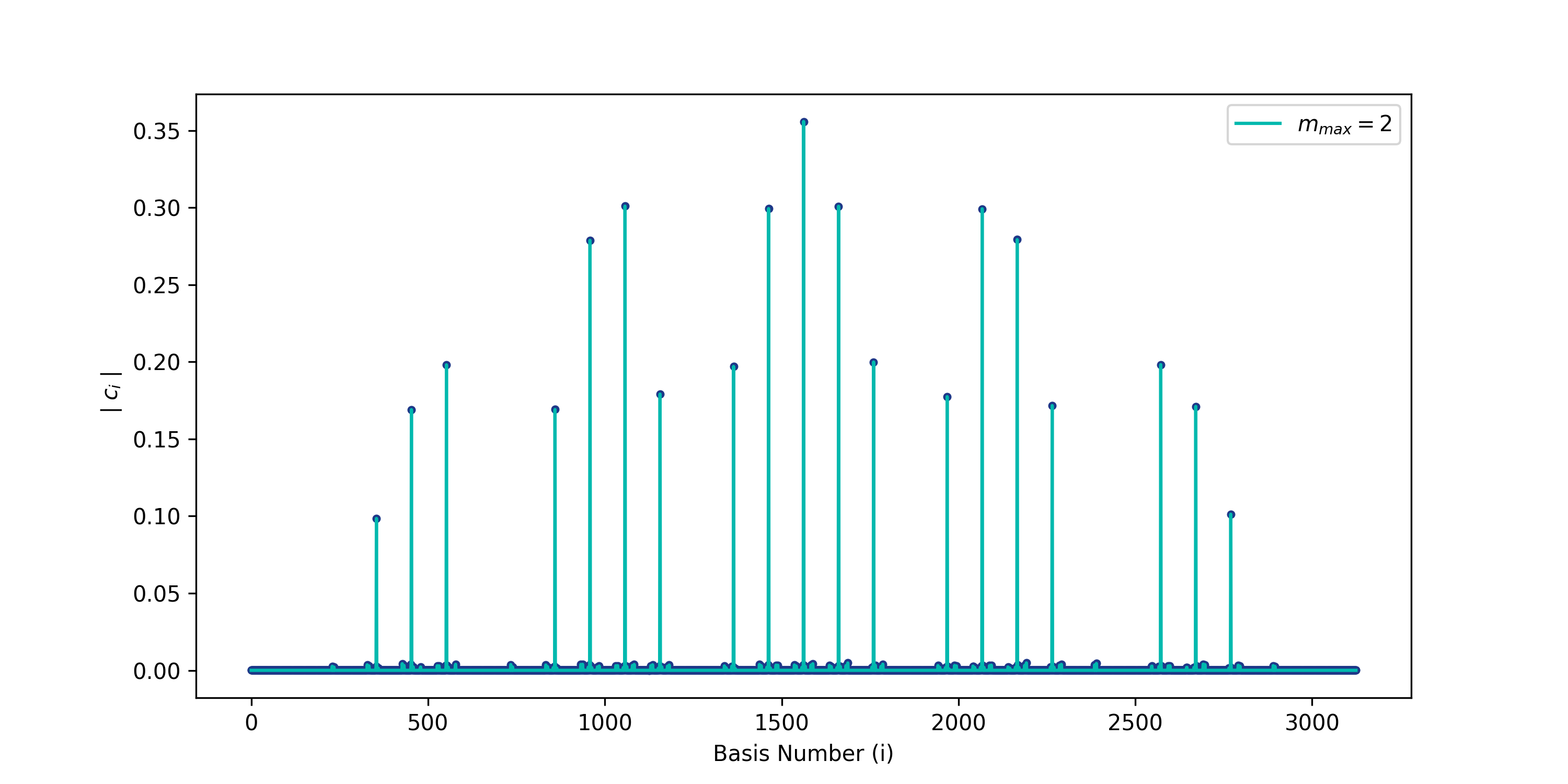}
    \caption{The amplitudes of the 3125 basis states obtained using the neural network in the $\jmax = 2$ cutoff. We see that only 19 of them (0.608 \% of the entire space) are strongly contributing.}
    \label{fig:spin2_contributingStates}
\end{figure}
\newline
Once the amplitudes of those strongly contributing basis states are identified, one can see precisely which basis states are those. It was observed that the strongly contributing states are all gauge invariant states. However, not all gauge invariant states contribute \textit{strongly}. Given the orientation of our graph, then the dimension of the gauge invariant Hilbert space would be $\dim\hilb_{\Tilde{\gamma}}^{G} = (2\jmax + 1)^2 = 25$ for $\jmax = 2$. On the contrast, only 19 of those contribute strongly in the solution obtained using the neural network. One can investigate precisely which of gauge invariant basis states are strongly contributing, however this is now graph orientation dependent, and thus is not conducted in this work.
\newline
Now we turn our focus into yet another test on the validity of the numerical simulations. For the work done so far, we considered a 2-L graph due to it being non-trivial. The triviality here is due to the fact that for a 1-L graph, the Gau{\ss} constraint is already perfectly imposed. However, if one considers a 1-L graph composed of three edges and three vertices, then the states describing our analytical model take the form
\begin{equation}
\label{eq:1LWavefunction}
    \Psi = \sum_{n = -\jmax}^{\jmax} c_n \ket{h_\alpha^{n}},
\end{equation}
where $\ket{h_\alpha^{n}}$ is a minimal loop holonomy eigenstate with charge $n$. In this simple case, the Gau{\ss} constraint is fully satisfied and the only remaining constraint is the curvature. Hence, the master constraint takes the form
\begin{equation}
    \hat{C} = (\hat{h}_\alpha - \one)(\hat{h}_\alpha^\dagger- \one) = -\hat{h}_\alpha - \hat{h}_\alpha^\dagger + 2\one.
\end{equation}
This constraint acting on the states shown in equation \eqref{eq:1LWavefunction} would result in
\begin{equation}
    \hat{C}\Psi = \sum_{n = -\jmax}^{\jmax}(2c_n - c_{n-1} - c_{n+1}) \ket{h_\alpha^n}.
\end{equation}
Written in matrix form, one can see that this is an almost tridiagonal matrix. Specifically, this is similar to a tridiagonal Toeplitz matrix where the entries on the diagonal are all equal to $2$, and the entries on the super and sub-diagonal are all equal to $-1$. We denote the diagonal entries by $d$ and the sub and super-diagonal entries by $d_{n-1}$ and $d_{n+1}$ respectively. Such matrices have a known eigenvalue spectrum and the corresponding set of eigenvectors. Specifically, the eigenvalues of an $N\times N$ tridiagonal Toeplitz matrix take the form 
\begin{equation}
    \lambda_k = d - 2\sqrt{d_{n-1}d_{n+1}} \cos\left(\frac{k\pi}{N+1}\right),
\end{equation}
where $i=1, \cdots , N$. Furthermore, the $l$\textsuperscript{th} component of the eigenvector $\Vec{v}_k$ corresponding to the eigenvalue $\lambda_k$ is given by
\begin{equation}
    v_{kl} = A \sin\left(\frac{kl\pi}{N + 1}\right),
\end{equation}
where $A$ is just a constant. Now, we consider a $\jmax = 2$ case. The point of interest is to verify the decay rate of any given amplitude in the solution when considered over different $\jmax$ cutoffs. To do this, we can take the norm squared of the ground state of a tridiagonal Toeplitz matrix after which it can be easily seen that $A = 1 / \sqrt{2\jmax + 1}$. In our case, this would mean that given a contributing basis state, its amplitude should contribute by a value which decays as $1 / \sqrt{2\jmax + 1}$ for higher $\jmax$ cutoff. This can be seen in Figure \ref{fig:decayRateOneLoop} where the chosen basis state was one with the charges on all three edges being zero.
\begin{figure}[h]
    \centering
    \includegraphics[scale=0.35]{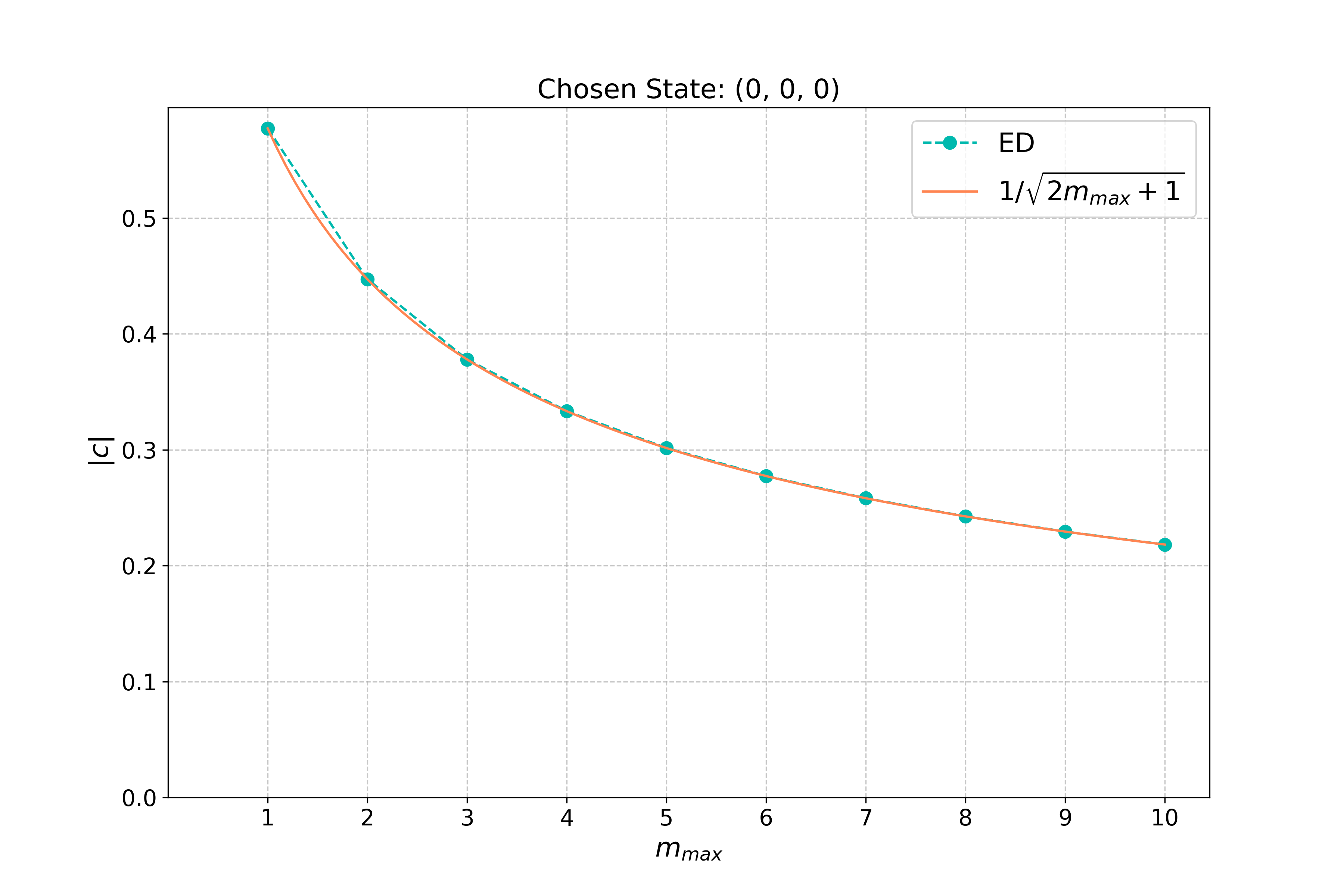}
    \caption{The decay rate of the amplitude of the state (0, 0, 0) in different $\jmax$ cutoffs is shown in light blue. These results were obtained from exact diagonalisation of the constraint which was defined on a 1-L graph. In orange, the decay rate obtained analytically is shown.}
    \label{fig:decayRateOneLoop}
\end{figure}
\newline
Figure \ref{fig:decayRateOneLoop} shows the amplitude of the contributing basis state (0, 0, 0) in different $\jmax$ cutoffs for the $\hat{C}$ defined on a 1-L graph. In light blue, we see the values of this amplitude obtained via exact diagonalisation methods. The orange line indicates the expected decay rate as dictated from the analysis of the ground state of a tridiagonal Toeplitz matrix conducted above. We see that indeed the amplitudes do follow the expected decay rate of $1 / \sqrt{2\jmax + 1}$ in the simple 1-L graph where the Gau{\ss} constraint is fully satisfied.
\newline
Conducting a similar study for a 2-L graph is not as straightforwards as the solution takes a more complex form as a result of the master constraint having a more complicated structure due to the Gau{\ss} constraint being not fully satisfied. In this case, it is difficult to get an exact analytical expression for the expected decay rate. However, using the analytical expression of the 1-L graph case as a guide, we see the results shown in Figure \ref{fig:cs_decay} where the decay of the amplitudes in a 2-L graph is nevertheless compared to the decay rate of $1 / \sqrt{2\jmax + 1}$ for a 1-L graph.
\begin{figure}[h]
\centering
\begin{subfigure}{.5\textwidth}
  \centering
  \includegraphics[width=1.05\linewidth]{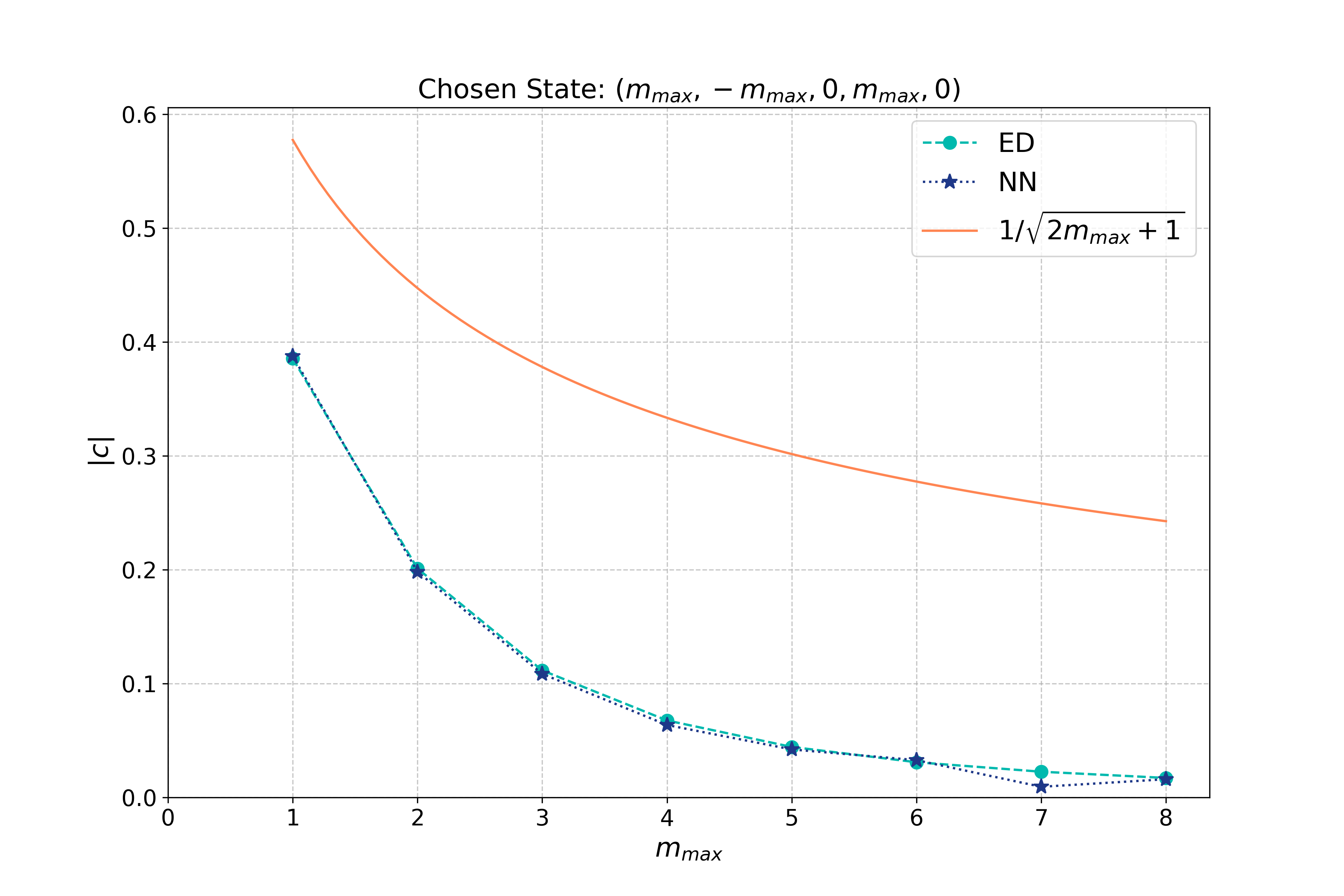}
\end{subfigure}%
\begin{subfigure}{.5\textwidth}
  \centering
  \includegraphics[width=1.05\linewidth]{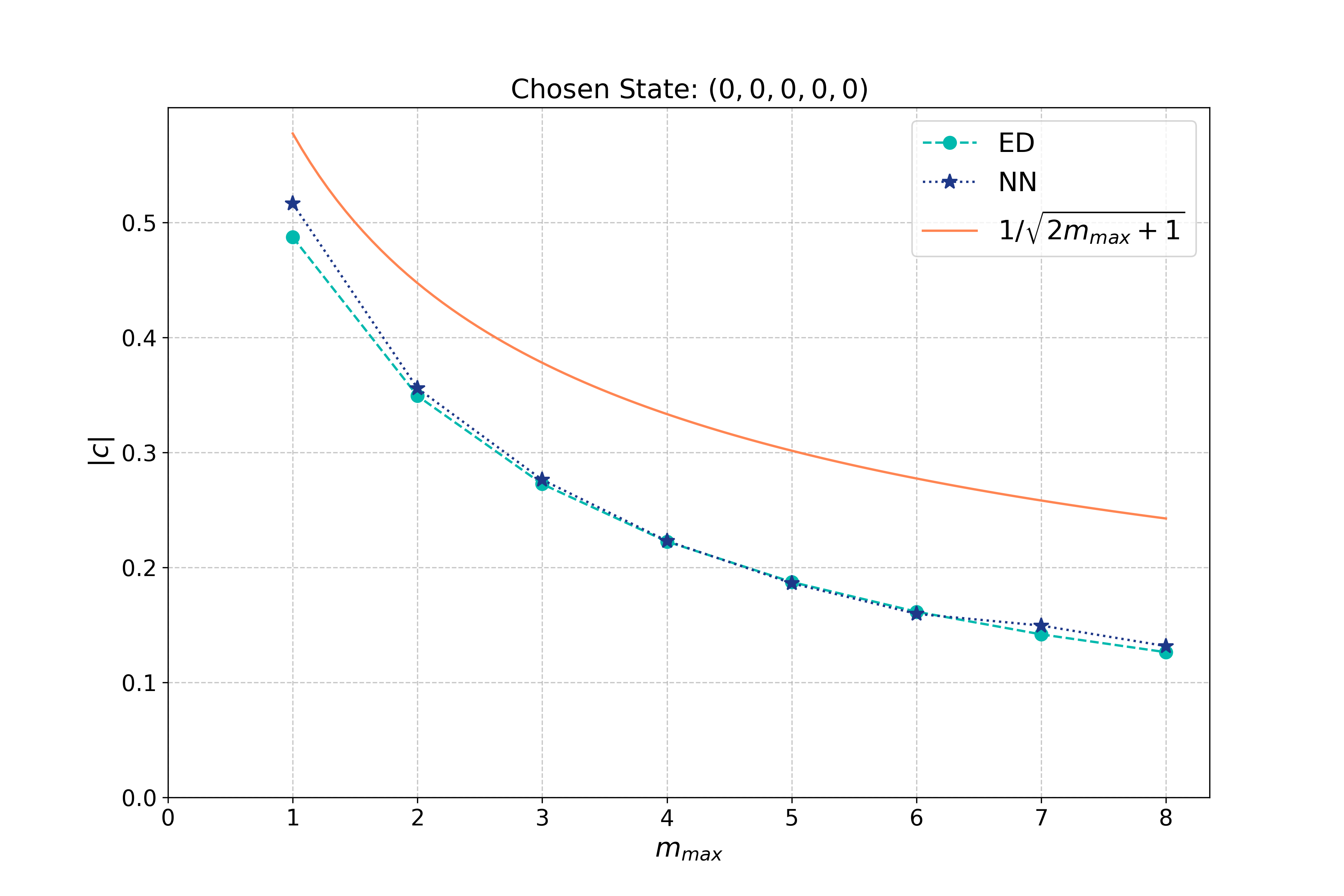}
\end{subfigure}
\caption{The decay rate of a chosen gauge invariant state in the solution of $\hat{C}$ for a 2-L graph in different $\jmax$ cutoffs is shown. The coefficients obtained from exact diagnalisation methods are shown in light blue and the coefficients obtained from the neural network are shown in violet. The orange line shows the decay rate of the amplitudes for the 1-L graph case.}
\label{fig:cs_decay}
\end{figure}
\newline
Figure \ref{fig:cs_decay} shows the decay rate of a chosen amplitude in the solution where the model is defined over a 2-L graph. On the left, the chosen state was one where one loop is \quotes{turned on} and the other loop is not (that is, one loop has zero charge on its edges which are not common to the other loop). On the right, the chosen state was one with all zero charges on all edges. Both the states are gauge invariant states which also exist in the 1-L graph. Their amplitudes are obtained using exact diagonalisation as well as using the neural network. The neural network results shown in violet match to a great degree the exact diagonalisation results shown in light blue. This further reflects the accuracy obtained in each of those simulations. Moreover, the solid orange line indicates the expected analytical decay rate for the 1-L graph case, as the similar expression for the 2-L graph case is more difficult to obtain. As seen in the figure, the amplitudes qualitatively follow the same decay rate. Further calculations show that the decay rate is faster than $1/\sqrt{2\jmax + 1}$ but slower than $1/(2\jmax + 1)$. 
\newline
It is a non-trivial task to get an analytical expression for the decay rate for an arbitrary N-L graph. Nevertheless, we have shown that the amplitudes of contributing states decay as analytically expected in the 1-L graph case and one does not see any reason why this should not be the case for N-L graphs where $(N\geq 2)$. This serves as yet another confirmation of the validity of the simulations.

\subsection{Entropy and R\'{e}nyi entanglement entropy in the solution of $\hat{C}$}
\label{sec:entropy}
In this section we present a brief discussion regarding the entropy and entanglement entropy in the solution. Note that the dimensions of the Hilbert space limit us from implementing von Neumann entropy for arbitrary $\jmax$ due to the growing amount of memory to be allocated for the density matrix. In this section we will present the Shannon entropy of the solution at different $\jmax$ cutoffs as well as some comments on the von Neumann entropy for a handful of different $\jmax$ simulations. Lastly, the entanglement entropy in the solution computed from the R\'{e}nyi entanglement entropy.
\newline
To begin, Shannon entropy, for a given set of probabilities, take the form
\begin{equation}
    S = -\sum_{i}p_i \log p_i.
\end{equation}
In our case, $p_i = |c_i|^2$. Since we have all the amplitudes, one can then it is straightforward to compute the Shannon entropy of the solution. In Table \ref{tab:entropy}, the results are summarised for different $\jmax$ cutoffs.
\begin{table}[h]

    \centering
    \begin{tabular}{c|cccccccc}
        \rowcolor{lightergreen!20}
        $\jmax$ & 1 & 2 & 3 & 4 & 5 & 6 & 7 & 8 \\
        \hline
        $S_{(\mathrm{NN})}$ & 0.33771 & 0.34421 & 0.34394 & 0.34302 & 0.34356 & 0.34783 & 0.33850 & 0.33881 \\
        $S_{(\mathrm{ED})}$ & 0.35018 & 0.34550 & 0.34318 & 0.34268 & 0.34327 & 0.34428 & 0.34541 & 0.34652 \\
    \end{tabular}
    \caption{The normalised Shannon entropy of the solution for different $\jmax$. The entropy is computed for the solution obtained via exact methods as well as the neural network.}
    \label{tab:entropy}
\end{table}
\newline
Table \ref{tab:entropy} shows the normalised Shannon entropy for the solution obtained via exact methods, denoted by $S_{(\mathrm{ED})}$ as well as using the neural network, denoted by $S_{(\mathrm{NN})}$. The entropies are normalised with respect to the maximum value the entropy can have in the respective Hilbert space. That is, if all the amplitudes are the same (thus \quotes{having the same probability}), the Shannon entropy reaches its maximum value of
\begin{equation}
    S = -\sum_{i}^{N} \left(\frac{1}{N}\log\frac{1}{N}\right) = \log N, 
\end{equation}
where $N=\dim\hilb_{\Tilde{\gamma}}$. This value is then used to normalise the Shannon entropy which in turn allows us to compare it across different Hilbert spaces with different $\jmax$. The minimum entropy is achieved when only one basis state contributes, where then $S = 0$. Analytically, the Shannon entropy is expected to be $\log\dim\hilb_{\Tilde{\gamma}}^{G} / \log\dim\hilb_{\Tilde{\gamma}}$, which in the given orientation of the graph is $\log(2\jmax + 1)^2 / \log(2\jmax + 1)^5 = 0.4$. What we see then from Table \ref{tab:entropy} is that the Shannon entropy, for both the exact methods and the neural network solutions, are very close to what one would expect them to be. 
\newline
The von Neumann entropy was computed only for $\jmax = 1, 2$ and $3$ due to being limited by the computational costs required for high $\jmax$ calculations. It was observed that the von Neumann entropy takes a very small value (orders of magnitude of $10^{-13}$). This rather small, almost zero, value for the von Neumann entropy goes as yet another consistency check in the books as one has a pure state at hand and as such one expects the von Neumann entropy to be zero.
\newline
The last part of this work concerns the entanglement entropy in the solution. Entanglement entropy is physical a quantity that measures the degree of quantum entanglement in an interacting many-body system. The manner in which the entanglement entropy scales can give insight into universal information of the system at hand. This can be used to characterise quantum phases of the system, and thus phase transitions in interacting models.
\newline
In this work, we aim to merely demonstrate the ability to compute such physically relevant quantities. We will consider the R\'{e}nyi entanglement entropy. Specifically, we will use the second R\'{e}nyi entanglement entropy which takes the form
\begin{equation}
    S_2(A) = -\log_2 \tr_A \rho^2.
\end{equation}
Here, $\tr_A$ indicates the partial trace over the partition $A$ of the system and $\rho = \ket{\Psi}\bra{\Psi}$ is the density matrix formed from the variational solution. We can consider the partition $A$ to be the entire graph, and hence we compute the R\'{e}nyi entanglement entropy for the entire state rather than a reduced state confined to a subregion of the system. It was observed that $S_2$ is almost zero, of magnitude $\approx 10^{-16}$ which is as in the case of the von Neumann entropy due to numerical precision errors. This further affirms that one has a pure state at hand as $S_2$ would vanish for a pure state. 
\newline
Furthermore, one can partition the system in different manners. We will opt to partition the system in 4 ways, each time computing $S_2$ accordingly. If we identify the partitions by the edges in the graph $\gamma$, then partitions are as follows
\begin{align}
    P_1 = \lbrace (0, 1), (0, 2), (0, 3) \rbrace \quad &, \quad P_2 = \lbrace (1, 2), (0, 2), (3, 2) \rbrace, \nonumber\\
    P_3 = \lbrace (0, 3), (3, 2) \rbrace \quad &, \quad P_4 = \lbrace (0, 1), (1, 2) \rbrace.
\end{align}
This can also be further elucidated upon by referring to the figure below whereby the partitions are shown on the graph $\gamma$. We note that edges included in any shaded partition region are meant to be included in that partition. 
\begin{figure}[h]
    \centering
    \includegraphics[scale=0.3]{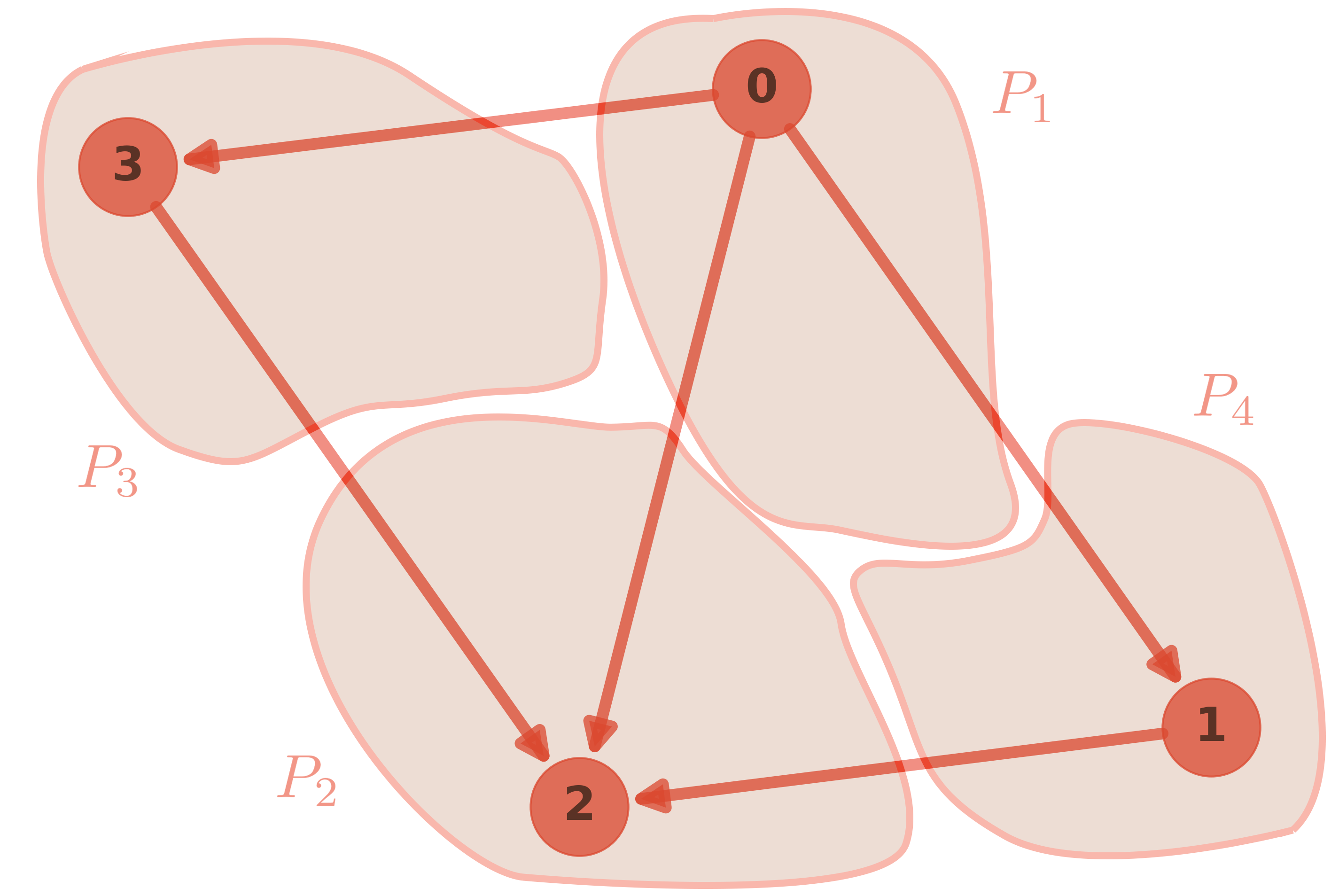}
    \caption{The four partitions chosen to compute the second R\'{e}nyi entropy $S_2$. As shown, two of them include 3-valent vertices ($P_1$ and $P_2$ respectively) while the remaining two are 2-valent vertices ($P_3$ and $P_4$).}
    \label{fig:partition}
\end{figure}

\noindent The reason for the partition is to isolate the different valent vertices. For example, $P_1$ includes all the edges connected to the 3-valent vertex labeled by 0 while $P_4$ is a partition which has all the edges attached to the 2-valent vertex labeled by 1 (see Figure \ref{fig:partition}).
\newline
It was observed that $S_2$ for any partition which included a 3-valent vertex had a higher value than when considering partitions with no 3-valent vertices. For example, for $\jmax = 1$, then $S_2(P_1) = 2.47 \pm 0.44$ while $S_2(P_4) = 1.25 \pm 0.24$. Furthermore, these values grow significantly for higher $\jmax$ where for $\jmax = 5$ then for the same partitions mentioned it was observed that $S_2(P_1) = 26.79 \pm 0.12$ while $S_2(P_4) = 4.06 \pm 0.81$. 
\newline
The findings obtained tentatively show that the second R\'{e}nyi entanglement entropy in the solution of the master constraint is related to the valence of the vertices. Determining whether the entanglement entropy follows an area or volume law is at the moment not possible as one does not have concrete quantum geometric observables in this quantum model with this simple gauge group and will be left for future work. We do note that the purpose of this brief study is not an extensive one. Rather, it is to demonstrate the capabilities of the softwares used and developed in this work. Therefore, conclusions from these results are only tentative and an exhaustive study of the entanglement entropy in relation to different partitions of the system is left for future work.

\section{Discussion and outlook}
In this work, we have considered a toy gravitational model in which we have demonstrated that it can be fully numerically solved using the neural network quantum state ansatz. Motivated by 3+1 gravity in the weak coupling limit, we have chosen the to consider a BF-theory in 3-dimensions with a Lie group of $\Uone$ in which the constraints took the form
\begin{equation}
    \extd\omega^I = 0 \quad , \quad \extd e^I = 0,
\end{equation}
where the $\omega^I$ and $e^I$ denote the spin connection and the orthonormal frame of 1-forms. This $\Uone$ BF-theory was quantised using loop quantum gravity methods. In doing so, one gains the advantage of being able to define the model on a graph which in turn allowed for an implementation of the $\Uone$ BF-theory on a computational basis. Each edge of the graph carried a charge $m \in \Z$ corresponding to the representation of the $\Uone$ holonomy on the given edge. In efforts to make the model more computationally tractable, the representation algebra of the labels of the holonomies was deformed to restrict the allowed charges to fall within an admissible set of charge numbers $M \subseteq \Z$. Due to this, we obtained what we called a quantum $\Uqone$ BF-theory.
\newline
The curvature constraint imposed flatness over the minimal loops of the graph and the Gau{\ss} constraint ensured charge conservation at every vertex. We have then considered a master constraint for the model which was expressed as
\begin{equation}
    \hat{C} = \sum_{\alpha \in L(\gamma)}\left(\hat{h}_\alpha - \one\right)\left(\hat{h}_\alpha^\dagger - \one\right) + \sum_{v \in V(\gamma)}\left(\sum_{e \in E_i(\gamma)} \hat{N}_{e} - \sum_{e' \in E_o(\gamma)}\hat{N}_{e'}\right)^2.
\end{equation}
The computational tools used in this work require the Hilbert space to be defined on vertices. Since the edges carry the charges, this consequently led to modelling the physical system on the dual graph in the computational basis. In there, the wave-function was that of a many-body system, where the total Hilbert space $\hilb_{\Tilde{\gamma}}$ over the entire dual graph was decomposed as a tensor product of Hilbert spaces $\hilb_{\Tilde{v}}$ defined over the dual vertices $\Tilde{v} \in \Tilde{\gamma}$. The basis states of $\hilb_{\Tilde{v}}$ were labeled by the charge numbers $m$.
\newline
The physical system was numerically solved using standard exact numerical techniques (Lanczos algorithm) to find the eigenvalues and eigenvectors of the master constraint $\hat{C}$. We arrived at the state with \quotes{lowest energy} (if one understands the constraint as the Hamiltonian operator) which is the solution for $\hat{C}$. Furthermore, the NNQS ansatz was employed to solve this model and proved to be of rather remarkable accuracy and efficiency. We showed that as the $\jmax$ cutoff was relaxed, the quantum fluctuations of the minimal loop holonomies died out and the non-commutativity between the curvature and Gau{\ss} constraint approached zero. It is then shown that one approaches the continuum limit quickly, even for relatively low $\jmax$. We showed that the basis states which contribute to the solution of the constraints are all gauge invariant states. Furthermore, the amplitudes of such states were verified to follow the expected decay rate as one approaches the continuum limit. The Shannon entropy of the solution was computed in which we showed that the results in every considered different $\jmax$ simulation followed the analytically derived results. Furthermore, the von Neumann entropy computed for the solutions of a handful of different $\jmax$ simulations further supported the fact that one deals with pure states. Lastly, we concluded the study by looking into the second R\'{e}nyi entanglement entropy $S_2$ where the system was partitioned to several partitions which included the edges surrounding the 2 and 3-valent vertices. It was observed that it appears as though the entanglement entropy is related to the valence of the vertices. However, such a statement is only tentative and due to the lack of geometric observables, the law which the entanglement entropy follows in the solution remains to be explored in future work.
\newline
This work stands as a proof of concept of the application of NNQS in gravity inspired models. It is important to address not only the results but the points of possible improvement in this work. For one thing, the architecture of the network used in this work for the NNQS ansatz was deliberately chosen to be not $\Uone$ gauge invariant. Such invariance under certain gauge groups for NNQS exist, and it would be interesting to compare the performance of the different networks. This especially becomes interesting in large graphs or high $\jmax$ cutoffs. Upcoming work \cite{Sahlmann:2024kat} shows that this also becomes very crucial when considering models with a different gauge group such as $\Uone^3$ instead of $\Uone$. Another point would be perhaps considering more physically rich toy models which include matter fields as a test-bed for future realistic models. 
\newline
Ultimately, the method presented here has to be generalized to the case of the non-Abelian group SU(2). Several challenges are expected on this path: 
\begin{itemize}
    \item To obtain a mathematically consistent cutoff, one would probably work with SU(2)$_q$ and its representations. 
    \item For a given maximal spin $j_{max}$, the Hilbert space dimension is bigger than in the corresponding $\Uqone^3$ case, as all irreps with $j\leq j_{max}$ can contribute. This is a performance concern. 
    \item Some or, depending on the formulation, all of the quantum states reside in nontrivial intertwiner spaces, so this structure has to be suitably described within the framework of NetKet. 
    \item The volume operator, an important ingredient in the Hamilton constraint for loop quantum gravity is quite complicated. We have successfully demonstrated the use of a volume operator in the $\Uqone^3$ case \cite{Sahlmann:2024kat}, but the one for SU(2)$_q$ is more complicated. 
\end{itemize}
In any case, we hope that this work is the first step in using modern numerical approaches in the domain of quantum gravity.

\ack
The authors are grateful for questions and comments by S.\ Steinhaus which helped improve the content and presentation. Conversations with B.\ Dittrich and T.\ Thiemann are also gratefully acknowledged. H.S. acknowledges the contribution of the COST Action CA18108. The authors gratefully acknowledge the scientific support and HPC resources provided by the Erlangen National High Performance Computing Center (NHR@FAU) of the Friedrich-Alexander-Universität Erlangen-Nürnberg (FAU). The hardware is funded by the German Research Foundation (DFG). The numerical simulations were carried out using Netket \cite{CARLEO2019100311}.

\clearpage
\section*{Appendices}
\renewcommand{\thesubsection}{\Alph{subsection}} 

\subsection{Brief outline of the optimisation process}
\label{app:A}
In \cite{Carleo:2017nvk}, it was first shown that a RBM can be used as a variational ansatz for a quantum state. For a system of $N$ spin particles, the amplitude of the many-body wave-function takes the form
\begin{equation}
    \Psi_{K}(\underline{\sigma}; a, b, w) = e^{\sum_i a_i \sigma^z_i}\prod_{j = 1}^{K}2\cosh\theta_j (\underline{\sigma}) \,\, , \,\, \theta_{j}(\underline{\sigma}) = b_j + \sum_{i}^{N} w_{ij}\sigma^z_i,
\end{equation}
where $a, b$ and $w$ are complex valued network parameters and $\underline{\sigma}$ denote the basis elements of the space. The training of the ansatz is a form of unsupervised training. Thus, one can train the ansatz with variational Monte-Carlo (VMC) \cite{Carleo:2017nvk,gubernatis_kawashima_werner_2016}. This can be done by observing that the expectation value of some given operator $\hat{O}$ is a function of the network parameters. Therefore, by computing the $\expect{\hat{O}}$, one can update the network parameters such that $\expect{\hat{O}}$ (the cost function in machine learning terms) is minimised. In what follows, we outline this process.

\subsubsection{Expectation values via Monte-Carlo}

Given that one deals with Hilbert spaces of very large dimensions, it would be inefficient and sometimes infeasible to compute the expectation value of an operator explicitly. Therefore, Monte-Carlo methods need to be employed. In such a case, the expectation value of an operator is given by \cite{PhysRevResearch.2.023358}
\begin{equation}
\label{eq:MCExpect}
    \expect{\hat{O}} = \sum_{i}^{N_{MC}} \frac{O_{loc}(\underline{\sigma}_i)}{N_{MC}},
\end{equation}
where here $N_{MC}$ is the total number of samples in the Monte-Carlo process (number of chains $L$ $\times$ number of samples per chain $S_L$), $\underline{\sigma}_i$ are the samples in the $i$\textsuperscript{th} chain and the \textit{local estimator} $O_{loc}$ is defined by \cite{PhysRevResearch.2.023358}
\begin{align}
    O_{loc}(\underline{\sigma}) & = \sum_{\underline{\sigma}'} \langle \underline{\sigma} |\hat{O}| \underline{\sigma}' \rangle \frac{\Psi_K(\underline{\sigma}')}{\Psi_K (\underline{\sigma})} \\
    & = \sum_{\underline{\sigma}' \text{ s.t. } \langle \underline{\sigma} |\hat{O}| \underline{\sigma}' \rangle \neq 0} \langle \underline{\sigma} |\hat{O}| \underline{\sigma}' \rangle \exp(\log \Psi_K (\underline{\sigma}') - \log \Psi_K(\underline{\sigma})).
\end{align}
The sum in the second line of the equation above is only on basis elements which have non-zero contributions hence reducing the computational demand. Given a number of samples drawn from a distribution $|\Psi_K(\underline{\sigma})|^2 / \sum_{\underline{\sigma}} |\Psi_K (\underline{\sigma})|$, the average energy over the samples gives an unbiased estimate of the energy. It can be easily seen that the computational cost is now dependent on how sparse $\hat{O}$ is instead of requiring the full dense matrix for $\hat{O}$.
\newline
Since one considers $L$ chains, that means one has $L$ evaluations for the local estimator $O_{loc}$. Error bars can then be computed as 
\begin{equation}
    \text{Err}. = \sqrt{\frac{\text{Var}(\lbrace O_{loc} \rbrace)}{N_{MC}}},
\end{equation}
where $\lbrace O_{loc} \rbrace$ is the set of $L$ evaluations mentioned prior.

\subsubsection{Generating samples in the chains}
The sampling process of samples to form chains is, as mentioned, done using some Monte-Carlo sampling. Specifically, one uses Markov Chain Monte-Carlo (MCMC) \cite{Carleo:2017nvk,Dellaportas2003} sampling. The set of configurations $\underline{\sigma}_1, \underline{\sigma}_2, \dots$, referred to as chains, are constructed using some Metropolis-Hasting type sampler or an exact sampler if the space is small enough. This sampler proposes a new sample, based on the current sample, using a specified transition kernel. After every iteration, it is either accepted or rejected based on an acceptance criteria, specifically \cite{Carleo:2017nvk}
\begin{equation}
    P(\underline{\sigma}_{k + 1} = \underline{\sigma}_{\mathrm{prop}}) = \min\left(1, \left|\frac{\Psi_K (\underline{\sigma}_{\mathrm{prop}})}{\Psi_K (\underline{\sigma})}\right|^2 \right),
\end{equation}
and the sample is then the set of configurations (chains) of the Markov chain down-sampled at an interval. Once this is done, the expectation value can be computed as outlined prior. What remains then is updating the network parameters to minimise the expectation value.

\subsubsection{Updating the network parameters}

In this work, the network parameters are updated in a two-step process. First, using a \textit{preconditioner} that transforms the gradient of the cost function, in this case $\expect{\hat{O}}$, in order to improve convergence properties before passing it to an optimiser such as a stochastic gradient descent (SGD) optimiser. In this work, stochastic reconfiguration (SR) is used \cite{Sorella_2007}. In such a method, variational derivatives with respect to the $M$\textsuperscript{th} network parameter are written as \cite{Carleo:2017nvk}
\begin{equation}
    \mathcal{O}_M(\underline{\sigma}) = \frac{1}{\Psi_K (\underline{\sigma})} \partial_{\mathcal{W}_M} \Psi_K (\underline{\sigma}),
\end{equation}
where $\mathcal{W}$ denote the weights of the network. SR updates, at the $p$\textsuperscript{th} iteration, the network parameters such that \cite{Carleo:2017nvk}
\begin{equation}
    \mathcal{W}_{p + 1} = \mathcal{W}_p - \gamma(p) S^{-1}(p) F(p),
\end{equation}
where $\gamma(p)$ is a scaling parameter, $S(p)$ is a (Hermitian) covariance matrix defined as \cite{Carleo:2017nvk}
\begin{equation}
    S_{K K'}(p) = \expect{\mathcal{O}^*_K \mathcal{O}_{K'}} - \expect{\mathcal{O}_K} \expect{\mathcal{O}_{K'}},
\end{equation}
and the forces $F(p)$ are given by \cite{Carleo:2017nvk}
\begin{equation}
    F_K(p) = \expect{O_{loc} \mathcal{O}_K^*} - \expect{O_{loc}}\expect{\mathcal{O}_K^*}.
\end{equation}
The next step is to compute the variational derivative of the cost function $\expect{\hat{O}}$ by computing the variational derivatives of equation \eqref{eq:MCExpect} after which the network parameters are updated such that \cite{PhysRevResearch.2.023358}
\begin{equation}
\label{eq:SGDUpdate}
    \mathcal{W}_{p + 1} = \mathcal{W}_{p} - \eta \partial_{\mathcal{W}_{p}}\expect{\hat{O}},
\end{equation}
if one uses SGD, where $\eta$ is the \textit{learning rate}. All derivatives are computed using automatic differentiation \cite{Zhang_2023}. In this work, we opted to use the Adam (Adaptive Momentum) optimiser \cite{kingma2017adammethodstochasticoptimization} as opposed to SGD. This was due to observing better convergence using the former. The Adam optimiser updates the network parameter in a similar albeit more complicated manner. The main difference is that SGD updates the parameters $\mathcal{W}$ based on the current gradient of the cost function as seen in equation \eqref{eq:SGDUpdate}. Adam, on the other hand, adjusts each parameter's learning rate individually. This is done based on estimates of both the first moment (mean) and the second moment (uncentered variance) of the gradients. This gives the Adam optimiser the ability to handle sparse gradients more effectively and reduces the probability of getting stuck in a local minima.
\newline
If we denote by $g_i := \partial_{\mathcal{W}_i}\expect{\hat{O}}$, then the first moment is computed as
\begin{equation}
    m_{i} = \beta_1 m_{i-1} + (1 - \beta_1) g_i,
\end{equation}
where $\beta_1$ is the decay rate of the first moment. The second moment can be then computed as
\begin{equation}
    v_{i} = \beta_2 v_{i-1} + (1 - \beta_2) g_i^2,
\end{equation}
where $\beta_2$ is the decay rate of the second moment. The bias correction for both moments can then be evaluated as
\begin{equation}
    \Tilde{m}_i = \frac{m_i}{1 - \beta_1^i} \qquad , \qquad \Tilde{v}_i = \frac{v_i}{1 - \beta_2^i}.
\end{equation}
All together, the parameter update rule according to the Adam optimiser, for a small parameter $\epsilon$ needed to avoid division by zero, takes the expression
\begin{equation}
    \mathcal{W}_{p + 1} = \mathcal{W}_p - \frac{\eta_p}{\sqrt{\Tilde{v}_p} + \epsilon} \Tilde{m}_p,
\end{equation}
which clearly has a richer structure when compared to the update rule of SGD in \eqref{eq:SGDUpdate}. Note that the learning rate $\eta_p$ has a dependence on the iteration $p$ here since in practice, one can provide a number of learning rates using a schedule function to the software used in this work.

\subsection{Comparing different samplers with different $\hat{R}$ values}
\label{app:B}
\subsubsection{General overview of sampler types}
The MCMC process requires a sampler to draw samples from the target distribution. The standard type samplers one would usually use are Metropolis-Hastings type samplers \cite{Chib1995}. Such samplers would sample proposals from a proposal distribution, evaluate the acceptance probability for the given proposals, and update the current state based on the acceptance probability. This iterative process is the general outline for exploring the target distribution and generating representative samples. 
\newline
On the other hand, one also have \textit{exact} samplers (exact inference methods) \cite{Brémaud2017}. Such samplers are used to obtain samples exactly from a target distribution. That is, unlike Metropolis samplers, exact samplers do not provide approximate samples, and they guarantee that the generated samples follow the true distribution. This can be done, for example, by enumeration where for a small enough discrete probability distribution, one can calculate the exact probabilities of each state. 
\newline
Ultimately, a sampler of some type is used to provide the chains. Thus, chains in an MCMC process are essentially sequences of the iteratively generated samples where each sample depends on the one preceding it. To indicate whether these chains are mixing well and are converging, the Gelman-Rubin statistic (denoted by $\hat{R}$) is used \cite{doi:10.1080/10618600.1998.10474787}. It compares the variance within each chain, to the variance of the means across chains. When the chains have converged to a stable distribution, then $\hat{R}$ approaches 1. In Metropolis type sampling, multiple chains are run in parallel, thus one has a $\hat{R}$ value for the simulation. However, exact sampling only involves one chain, thus $\hat{R}$ is not available in this case. Furthermore, because exact sampling does not have multiple chains, they do not suffer from autocorrelation issues where the sucessive samples in a chain are not independent, but rather correlated. This type of correlation can lead to inefficient exploration of the probability space and slow convergence. 

\subsubsection{The issue with Metropolis type samplers}
In our work, several samplers have been used. First, a \textit{local} Metropolis sampler where, given a basis state represented by charge numbers on every site, it will change the value of one random site to a random charge value from the allowed charges. Second, a \textit{2-local} Metropolis sampler, which acts like the local Metropolis sampler except it changes the charges for 2 random sites. Third, an \textit{exchange} Metropolis sampler, which exchanges the charges on two sites which are separated by a distance $d$. This type of sampler preserves the overall total charge of the state, as it only switches the values in the state around. Lastly, an exact sampler was also used.
\newline
The exact sampler acted as the benchmark which one could compare the results of the other samplers to. The exchange Metropolis sampler was immediately discarded, as it does not lead to effective exploration of the space due to it conserving the total charge. This property is not needed in our model. The local and 2-local Metropolis samplers gave similar accuracy, which were also close to the exact sampler. However, both suffered from a high $\hat{R}$ value. For the local sampler case, it is possible to find appropriate sampler parameters such that a low $\hat{R}$ value is obtained. However, this becomes more challenging as the charge cutoff is increased. Therefore, one would be inclined to dismiss these types of samplers, and only rely on the exact sampler. Unfortunately, exact samplers would only be available for small systems (e.g. a relatively small graph with a small charge cutoff) and thus one ultimately needs to use a Metropolis type sampler.

\subsubsection{Comparing the local Metropolis and the exact sampler}
In the following, we justify the use of the Metropolis type samplers despite the high $\hat{R}$ value. This is done by comparing the data across different samplers and identifying precisely why the Metropolis samplers fail and showing that this does \textit{not} affect the results obtained. The graph shown in Figure \ref{fig:graph} was used and a $\jmax = 2$ cutoff was imposed. The system was solved using the exact sampler (ES), a \quotes{good} (low $\hat{R}$ value) local Metropolis sampler (G-ML), a \quotes{bad} (high $\hat{R}$ value) local Metropolis sampler (B-ML) and a bad 2-local Metropolis sampler (B-MTL). The following table shows the results of the samplers compared to exact diagonalisation as well as their $\hat{R}$ values.
\begin{table}[h]
    \centering
    \begin{tabular}{c|cccc}
    \rowcolor{lightergreen!20}
        Sampler & $\min\expect{\hat{C}}$ & Accuracy (\%) & $\hat{R}$ & $\hat{R}_{\Bar{100}}$ \\
        \hline
        ES & 0.632 ± 0.015 & 94.788 & - & -  \\
        G-ML & 0.674 ± 0.012  & 87.889 & 1.0035 & 1.014  \\
        B-ML & 0.621 ± 0.027 & 96.619 & 1.2247 & 1.225  \\
        B-MTL & 0.743 ± 0.054 & 76.399 & 1.2247 & 1.225  \\
        \hline
        \multicolumn{5}{c}{\cellcolor{lightergreen!20} Exact diagonalisation: $\min\expect{\hat{C}} \approx 0.601165$}
    \end{tabular}
    \caption{The $\min\expect{\hat{C}}$ for a $\jmax = 2$ cutoff and the graph described in Figure \ref{fig:graph} obtained via the NNQS ansatz using different types of samplers: exact sampler (ES), good (low $\hat{R})$ local Metropolis (G-ML), bad (high $\hat{R}$) local Metropolis (B-ML), and a bad 2-local Metropolis sampler (B-MTL). The last column on the right indicates the average $\hat{R}$ value over the last 100 iterations of the simulation.}
    \label{tab:sampler_comparison}
\end{table}
\newline
We see from the Table \ref{tab:sampler_comparison} that all the samplers have a relatively close accuracy. The difference between a good and a bad local Metropolis sampler run is the sampler configuration (e.g. number of chains and so forth). Since the space is small enough, and one has the amplitudes for all the basis states at hand after the minimisation process, one can indeed investigate whether all of these samplers lead us to the same solution. More precisely, are the contributing states in every solution obtained in every sampler run the same. We see that this indeed is the case as shown in Figure \ref{fig:sampler_comparison}.
\begin{figure}
    \centering
    
    \begin{subfigure}{0.48\textwidth}
        \includegraphics[width=1.1\linewidth]{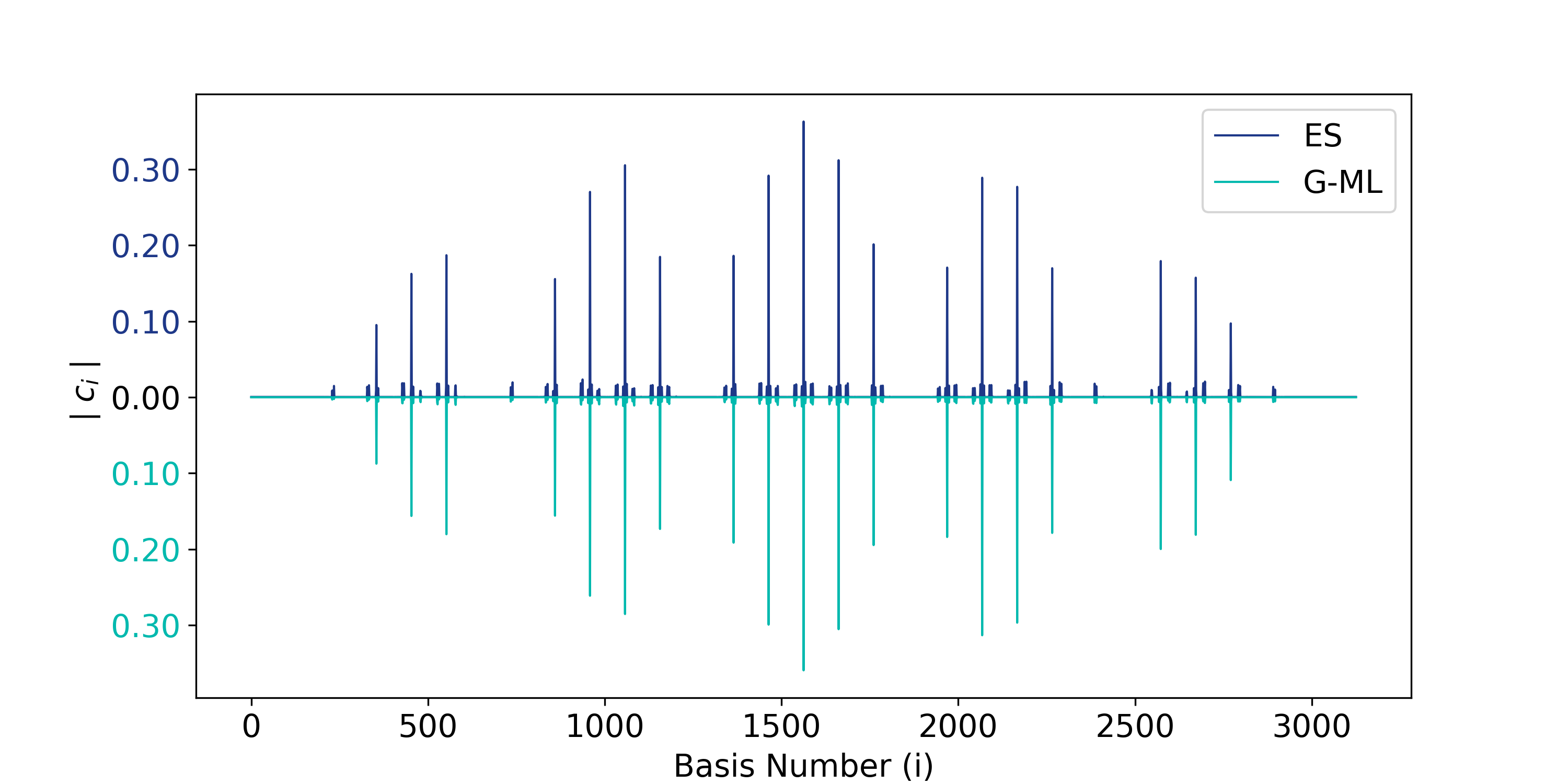}
        \caption{}
        \label{fig:sampler_comparison_sub1}
    \end{subfigure}
    \hfill
    \begin{subfigure}{0.48\textwidth}
        \includegraphics[width=1.1\linewidth]{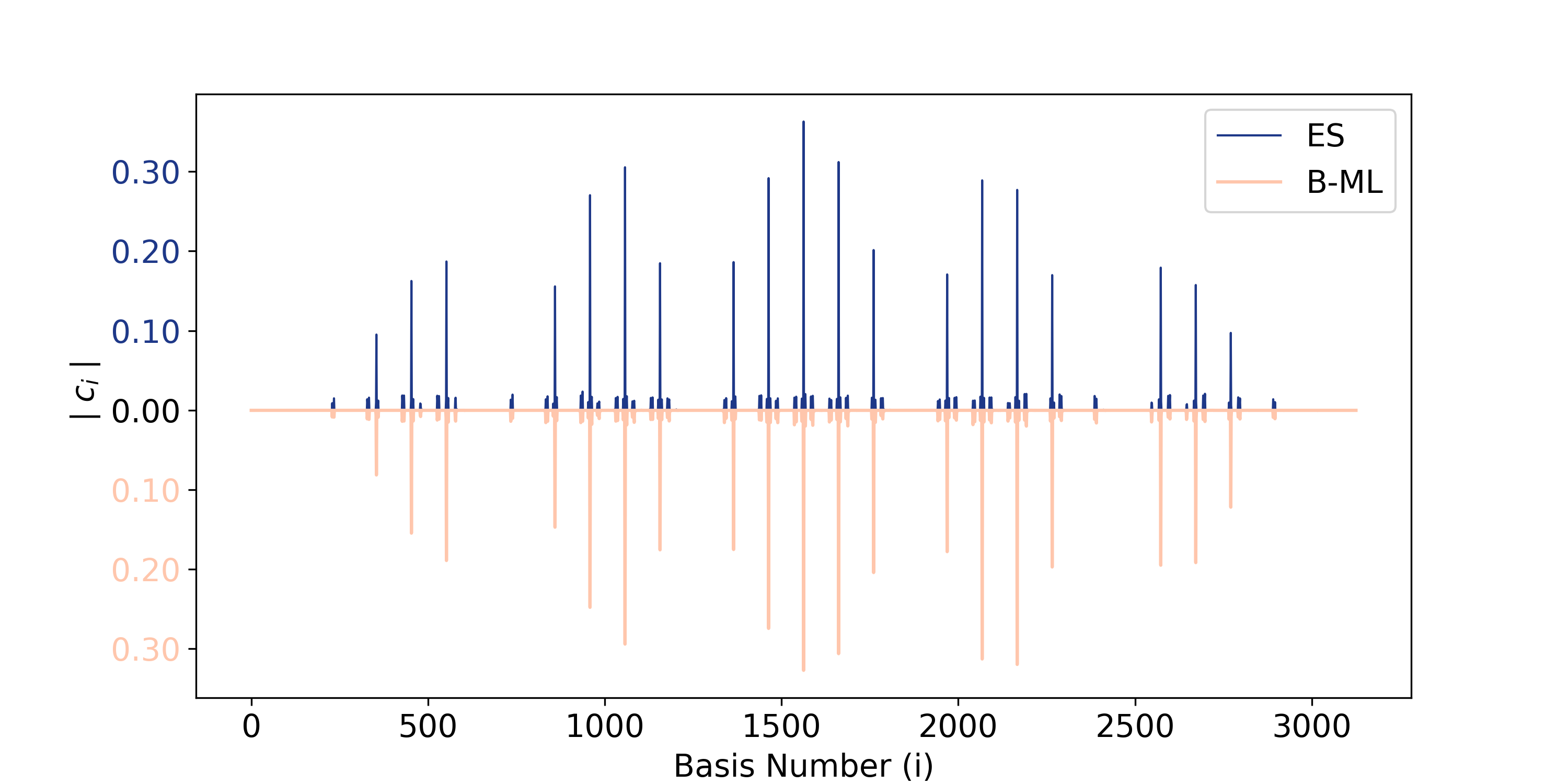}
        \caption{}
        \label{fig:sampler_comparison_sub2}
    \end{subfigure}
    
    \vspace{5px}
    
    \begin{subfigure}{0.48\textwidth}
        \includegraphics[width=1.1\linewidth]{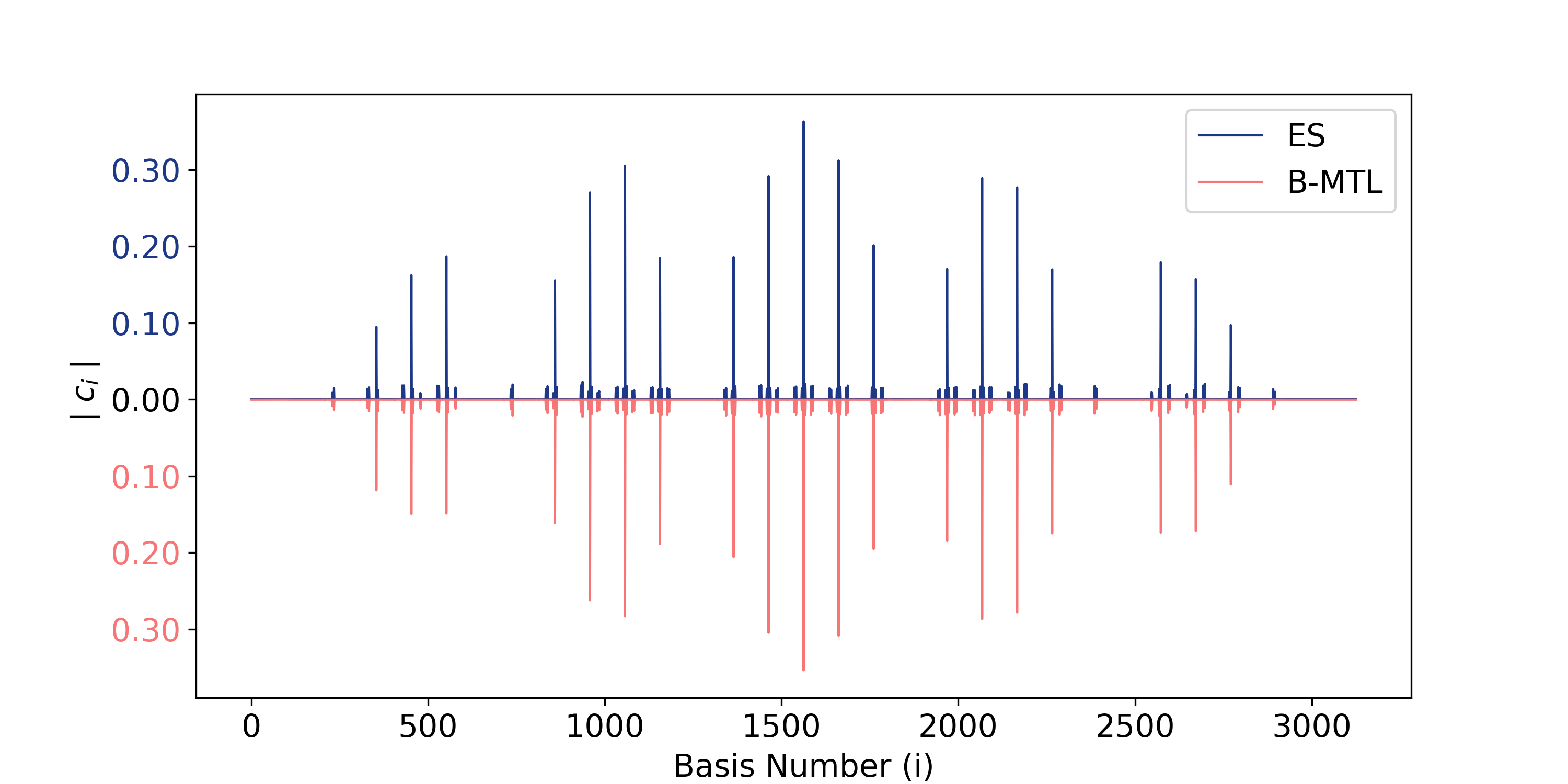}
        \caption{}
        \label{fig:sampler_comparison_sub3}
    \end{subfigure}
    
    \caption{The amplitudes of the solution in a $\jmax = 2$ simulation obtained using the NNQS ansatz with different samplers are shown. The exact sampler is compared with a Metropolis Local sampler with low and high $\hat{R}$ diagnostic as shown in Figure \ref{fig:sampler_comparison_sub1} and Figure \ref{fig:sampler_comparison_sub2} respectively and a Metropolis 2-local sampler with a high $\hat{R}$ diagnostic as shown in Figure \ref{fig:sampler_comparison_sub3}.}
    \label{fig:sampler_comparison}
\end{figure}
\newline
From Figure \ref{fig:sampler_comparison}, we see that despite the Metropolis type samplers failing and having a high $\hat{R}$ statistic, we arrive at the same solution as obtained by the exact sampler. This was furthermore numerically verified by computing the inner product of the states obtained where the \quotes{worst} case was $\langle \Psi^{(\mathrm{ES})} | \Psi^{(\mathrm{G-ML})} \rangle^2 = 0.98681$. These results are encouraging as one can be sure that despite the high $\hat{R}$ value, the state converges to the correct solution. The next task is to identify why the Metropolis type samplers fail, and why do they perform well under specific sampler parameters. 
\newline
One observes from the above figure that there are very specific, and not plenty, basis states that strongly contribute to the solution. The simulations are then performed for various $\jmax$ values, and the following table summarises the number of strongly contributing basis states in each case. 
\begin{table}[h]
    \centering
    \begin{tabular}{c|ccc}
    \rowcolor{lightergreen!20}
        $\jmax$ & $\dim\hilb_{\Tilde{\gamma}}$ & No. of contributing states & Percentage of the space (\%) \\
        \hline
        1 & 243 & 7 & 2.881 \\
        2 & 3125 & 19 & 0.608 \\
        3 & 16807 & 37 & 0.22  \\
        4 & 59049 & 61 & 0.103  \\
        5 & 161051 & 91 & 0.057 \\
        6 & 371293 & 127 & 0.034 \\
        7 & 759375 & 169 & 0.022 \\
        8 & 1419857 & 217 & 0.015 \\
    \end{tabular}
    \caption{The number of \textit{strongly} contributing basis states in the solution of different $\jmax$ cutoffs for the graph shown in Figure \ref{fig:graph}. The exact sampler was used for the MCMC process.}
    \label{tab:contributing_states_diff_samplers}
\end{table}
\newline
As shown in Table \ref{tab:contributing_states_diff_samplers}, the number of \textit{strongly} (See Section \ref{sec:contributingStates}) contributing basis states is relatively small to begin with, and gets smaller as $\jmax$ gets larger. This can also be seen as an example when visualising the amplitudes for the $\jmax = 5$ case as shown in Figure \ref{fig:spin_5_case}.
\begin{figure}[h]
    \centering
    \includegraphics[scale=0.45]{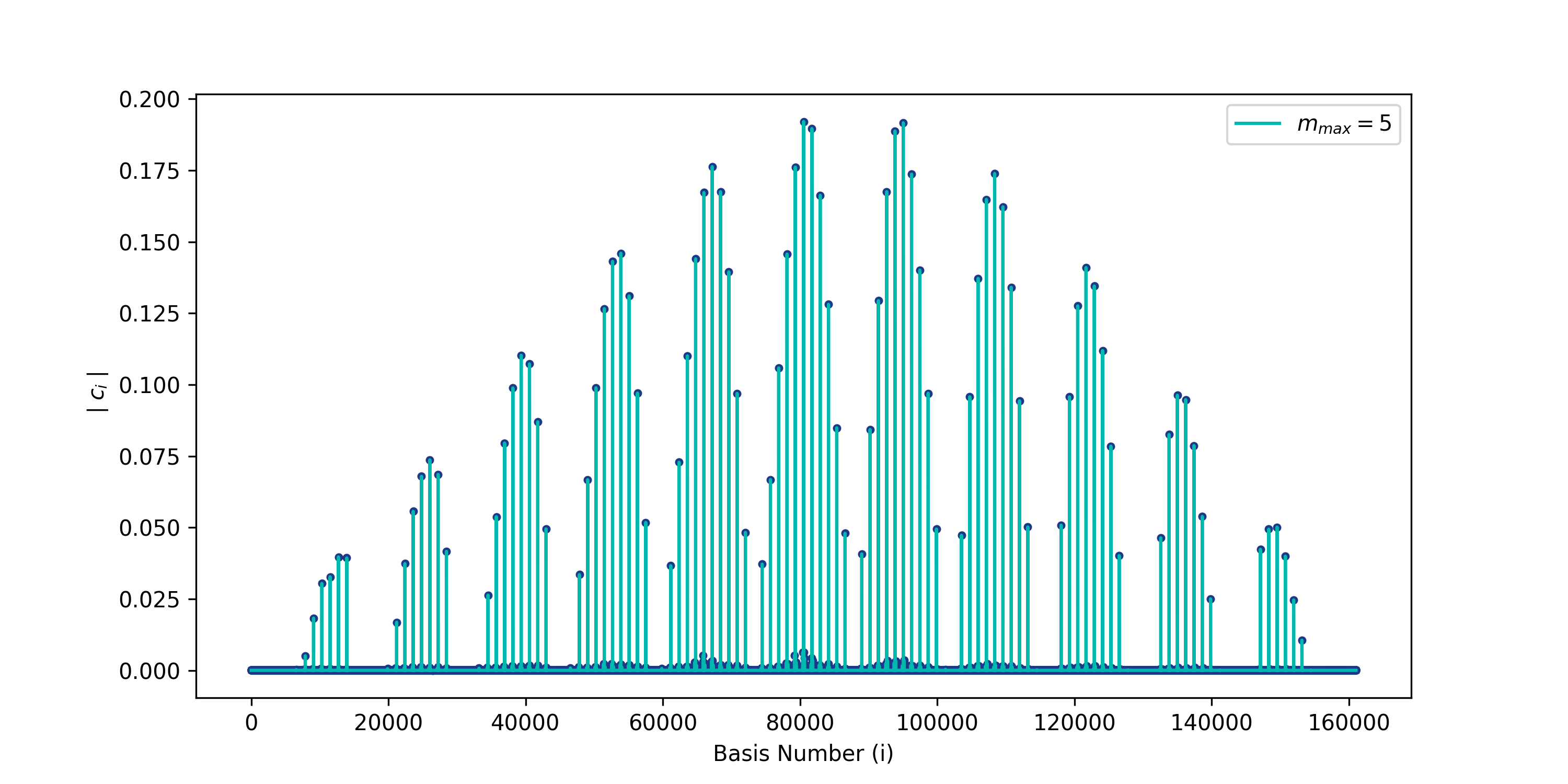}
    \caption{The amplitudes of all 161051 basis states in the $\jmax = 5$ cutoff for the graph shown in Figure \ref{fig:graph}. Out of them, only 91 are strongly contributing (0.057\% of the entire space). The mean of the amplitudes here is in the order of $10^{-5}$.}
    \label{fig:spin_5_case}
\end{figure}
\newline
What can be understood from these results is that as the charge cutoff gets higher, the number of strongly contributing states increases but not drastically especially in relation to the dimensions of the Hilebrt space. This in turn makes writing the solution in the computational basis a difficult task. Therefore, the Metropolis type samplers struggle to propose different proposals which will not be rejected by the MCMC. This is especially true when the number of chains in the sampler is set to a large value but with a relatively low number of sweeps per chain. The results obtained in the G-ML case were achieved by choosing a relatively low number of chains and a high sweeping rate. Therefore, the failure in the Metropolis type samplers, and thus the high $\hat{R}$, should \textit{not} be taken as an issue as this is inherit to the general type of problem at hand. This is further supported that even a good Metropolis sampler will start to have a bad $\hat{R}$ as $\jmax$ gets higher which is in accordance in the growing sparsity of the strongly contributing states for higher $\jmax$. While one can try to find good sampler configurations for different system sizes, ultimately, all samplers converge to the correct solution and the Metropolis type samplers can be used with confidence when the exact samplers cannot be implemented even if they have a large $\hat{R}$ value. This issue is not unique to this model. For example, in the Heisenberg model, the amplitudes in the solution are nicely spread out, while in the Bosonic matrix model, the case is as ours. Generally, it is observed that Metropolis type samplers will fail when the solution has an overlap with a very few number of basis states.

\subsection{Quantum states with real coefficients}
\label{app:C}
To increase numerical efficiency, it can be advantageous to do calculations using only real coefficients. Many calculations in this work are done in this way. Of course, that means that one a priori then explores a small part of the full Hilbert space of states. In this appendix, we will show that under certain circumstances, this subspace will have non-zero intersection with each eiegnspace of a given operator. In particular, in such a situation there would be a ground state with only real coefficients.
\newline
Let $\hilb$ be a complex Hilbert space and ${\{b_i\}}$ an orthonormal basis. We define complex conjugation with respect to this basis, i.e., for $\psi\in \hilb$ 
\begin{equation}
    \overline{\psi}:= \sum_i \, \scpr{\psi}{b_i}\, b_i. 
\end{equation}
We can then also define $\re \psi := (\psi +\overline{\psi})/2$ etc. in the obvious way. 
\begin{lemmaApp}
Let $A$ be an operator on $\hilb$ such that its matrix elements 
\begin{equation}
\label{eq:realmatrix}
    A_{jk}:= \scpr{b_j}{A\,b_k} 
\end{equation}
are real. Then each eigenvector is a (complex) linear combination of eigenvectors to the same eigenvalue with real coefficients in the basis $\{b_i\}$. In particular, for each eigenvalue $\lambda$ of $A$, there is an eigenvector $\psi_\lambda$ with real coefficients in the basis ${\{b_i\}}$, i.e. 
\begin{equation}
    \re \psi_\lambda = \psi_\lambda. 
\end{equation}
\end{lemmaApp}
\begin{proof}
Let $A$ be as stated, with non-zero eigenvector $\psi_\lambda$ to eigenvalue $\lambda\in\R$. Because of \eqref{eq:realmatrix}, 
\begin{equation}
    A\,\overline{\psi_\lambda} = \overline{A\,\psi_\lambda}= \lambda \overline{\psi_\lambda},
\end{equation}
so $\overline{\psi_\lambda}$ is an eigenvector to eigenvalue $\lambda$ as well. Consequently so are $\re \psi_\lambda$ and $\im \psi_\lambda$. This demonstrates the first part of the lemma. If one of them is nonzero, we have found the desired eigenvector with real coefficients. But one of them must be non-zero, or else $\psi_\lambda=0$ in contradiction to the assumption. 
\end{proof}

\newpage

\subsection{Technical notes on the network architecture and the computational resources}
\label{app:D}

\subsubsection{The network architecture}

The general network architecture used in this work and described in Figure \ref{fig:ourArchitecture} is detailed upon here. The following figure details the architecture in the $\jmax = 2$ case.
\vspace*{10px}
\begin{wrapfigure}{r}{5.5cm}
\includegraphics[width=5.5cm]{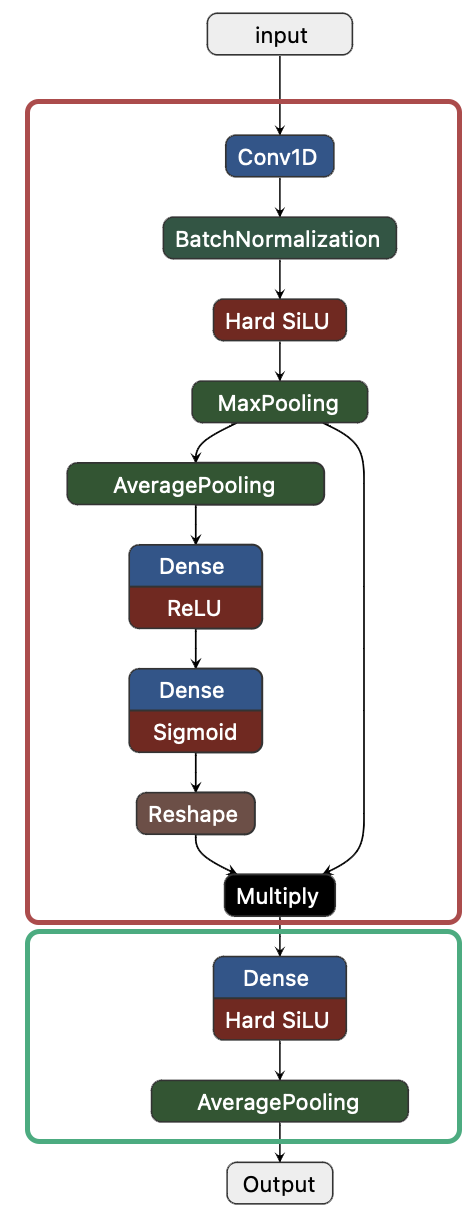}
\caption{The architecture of the neural network used in this work for the specific cutoff of $\jmax = 2$.}
\label{fig:explicitArchitecture}
\end{wrapfigure} 
\noindent
In the figure, the learning block is bounded by the red square while the evaluation block at the bottom is bounded by a green one (see Section \ref{sec:networkArchitecture}). In this specific case of $\jmax = 2$, there is only one convolutional sub-block which includes all the layers shown in the learning block. This sub-block composed of two parts: a typical convolutional network followed by a form of attention layer, specifically a squeeze-and-excitation (SE) block. In the former, the input is passed into a convolutional layer composed of a 1-D CNN with a kernel size of, for example, (1,) and 60 features. Next, a batch normalisation layer follows with a momentum (decay rate for the exponential moving average of the batch statistics) of 0.85 and a small value $\epsilon$ added to the variance to avoid dividing any values by zero. The following layer is an activation layer which utilises the Hard SiLU function
\begin{equation}
f(x) = \begin{cases}
            x \cdot \sigma(x), & \text{if } x > \text{threshold} \\
            0, & \text{if } x \leq \text{threshold}
        \end{cases}    
\end{equation}
where the sigmoid activation $\sigma(x)$ is 
\begin{equation}
    \sigma(x) = 1/(1 + e^{-x}).
\end{equation}
Lastly, a pooling (max. pooling) layer follows with a window shape of, for example, (1,). The next few layers and operations are the SE block. The use of an SE block is employed to enable the network to emphasise and focus on important features extracted by the previous CNN. Namely, the input is first pooled using average pooling after which it passes through a dense layer with some reduction ratio, in this case 2. That is, if the number of inputs are 60 then this dense layer will have only 30 output nodes. Following the activation by a ReLU function, which is defined as $f(x) = \max(x, 0)$, another dense layer follows. This time, it is comprised of as many nodes as there were previously input nodes prior to the reduction (e.g. 60). The output is then reshaped and multiplied by the output of the max. pooling layer. All of these layers mentioned so far comprise one convolutional sub-block. In the case of higher $\jmax$, the number of such sub-blocks, as well as the number of features in the convolutional layers and other variables changes to accommodate. The last part of the network, shown in green, is just a simple feed-forward network comprised of a single dense layer, an activation layer and lastly an average pooling layer.
\newline
While the architecture may look complicated, one of the key features which enables it to be robust is the convolutional layers. This is due to the fact that it has spatially local structures whereby the convolutional layer sweeps through the input using a kernel of a specified size. If a larger kernel is used, this enables the learning of broad features. On the other hand, a smaller sized kernel enables the learning of localised features. In the case at hand, that would heuristically mean that one relates the i\textsuperscript{th} quantum number in the \textit{input array} with the $(i+1)$\textsuperscript{th} quantum number. 
\newline
Note that the choice of such a small kernel size of (1,) for the convolutional layer mentioned earlier is intentional and in-effect does not take into account any such spatial correlations. This makes it equivalent to a dense layer with, in this case, 60 features. While in this work this is the case, the kernel size can be changed for larger spaces as needed and therefore the convolutional layer is the one used irrespectively. Further, the window shape of (1,) in any pooling layer in-effect leaves the input unchanged and does not aggregate any values. However, when increasing the kernel size of the convolution layer, this can also be changed and therefore the layer is left as it is in the general architecture, despite it not being of use in some cases, for generality as it does not add any number of parameters to the network.
\newline
The network accepts as input an array of dimensions \texttt{(batchdim, N)} where \texttt{batchdim} is the number of batches of basis states of the Hilbert space and \texttt{N} = $|E(\gamma)|$. In the most simplistic picture, each node in the input layer can be understood as the quantum number (representation label) of edges of the graph. Therefore, the general network architecture is agnostic to the graph used. However, for larger graphs a scaling similar to the one employed for large Hilbert spaces needs to be used. We do note that at the present time of writing, there is no ``automated" manner in which this scaling can take effect for \textit{any} graph and/or $\dim\hilb$. Further, for high $\jmax$, it was observed that implementing skip connections between different convolutional sub-blocks may aid with convergence. This work did not focus on implementing the most efficient architecture however this will be explored in upcoming work. 

\subsubsection{Computational resources}
This work utilised two computing devices to conduct the simulations. First, a high-performance computing (HPC) facility and second, a device running the Apple M1 (2020) processor. The following table summarises the technical details of both devices.
\newline
\begin{table}[h]
\small
    \centering
    \begin{tabular}{c|ccccccc}
    \rowcolor{lightergreen!20}
        Facility & Processor & Nodes & Cores & Base Freq. & Max. Freq. & GPU & MPI \\
         \hline
         HPC & Intel Xeon E3-1240 v6 & 1 & 4 & 3.70 GHz & 4.10 GHz & No & No \\
         Apple & Apple M1 (2020) & 1 & 8 & 3.20 GHz & N/A & iGPU & No \\
    \end{tabular}
    \caption{The technical specifications of the computational resources utilised in conducting the simulations done in this work.}
    \label{tab:my_label}
\end{table}

\noindent On both devices, no active parallelisation (MPI) has been utilised. That is, explicit distribution of the Markov chains or samples across different nodes or processors has not been done. Further, the HPC computations did not utilise any GPU. While the Apple M1 chip features an integrated GPU (iGPU), the softwares used in this work rely on, for the bulk of the computations, the jax \cite{jax2018github} package which is an XLA (Accelerated Linear Algebra \cite{50530}) compiler that is incompatible with the iGPU of the M1 chip. Therefore, no GPU has been used in the case of the M1 computations as well. 
\newline
The approach followed for the comparison of the computational cost in this work is an optimistic one. That is, we compare the best performance for exact diagonalisation (ED) computational costs against that of the neural network quantum state (NQS). The ED computations were superior on the HPC. This is attributed to several reasons such as the fact that the Xeon processor supports AVX2 (Advanced Vector Extensions) instructions optimised for high-performance numerical computations, optimisation for high double-precision (FP64) performance and hyper-threading technology to name a few. On the other hand, NQS simulations were superior on the M1 chip due to them being conducted on 4 performance cores. However, there are other contributing factors such as an increased memory bandwidth (approximately 68 GB/s compared to approximately 38 GB/s) and the Apple neural engine tasked to accelerate machine learning tasks. 
\newline
\begin{figure}[h]
    \centering
    \includegraphics[scale=0.55]{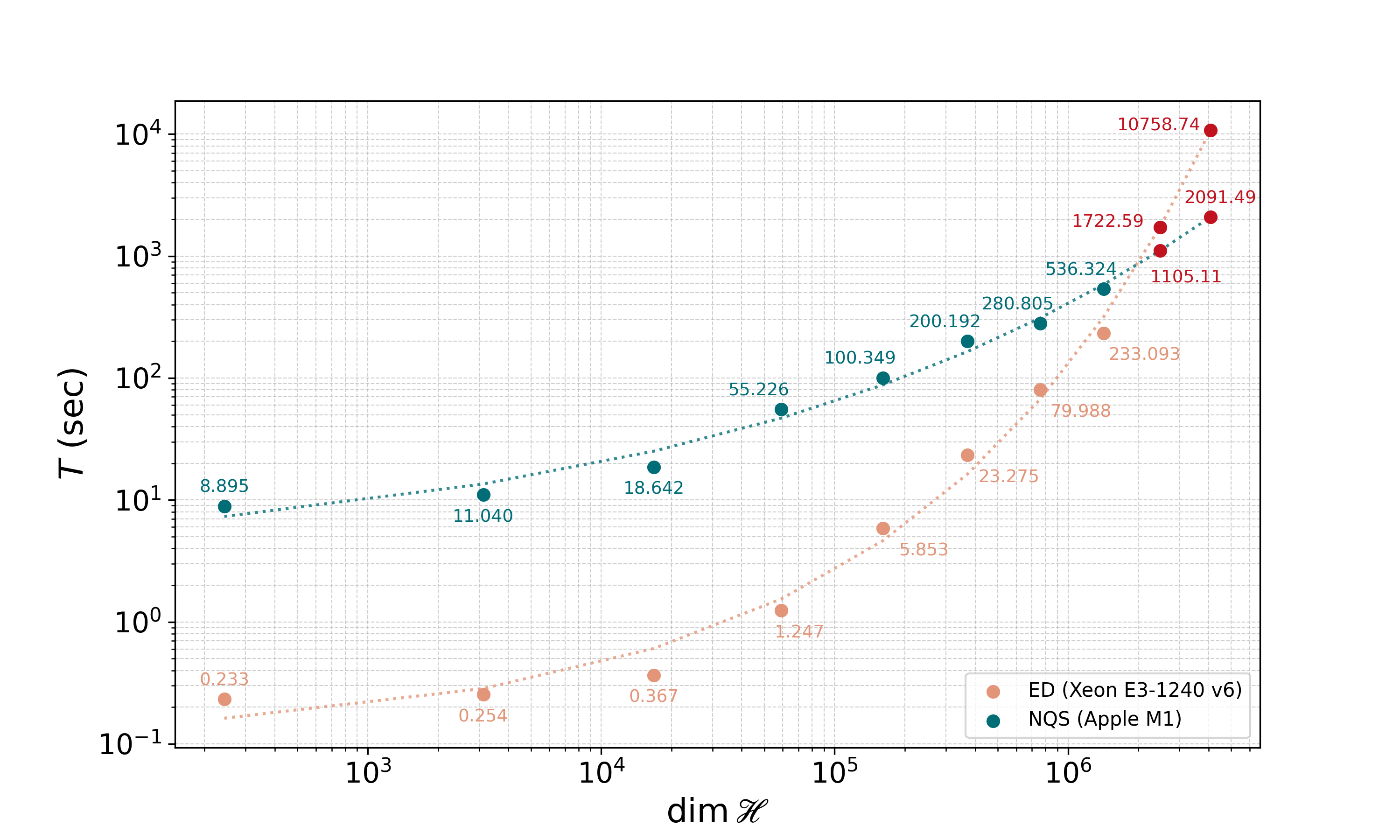}
    \caption{A log-log plot of the average runtime $T$ for ED and NQS computations (averaged over 10 runs) with respect to the dimensions of the Hilbert space. Note that the data points in red are extrapolated points. The data is fitted with the dotted line for each trend respectively.}
    \label{fig:computeTime}
\end{figure}

\noindent Figure \ref{fig:computeTime} shows the average runtime for conducting the ED and the NQS computations. Each data point represents the average of 10 simulations. The red data points in each trend are extrapolated from the fit function which is drawn as a dotted line for each dataset. The data obtained is fitted with the function
\begin{equation}
    T(\dim\hilb) = \gamma \exp(\alpha (\dim\hilb)^{\beta}).
\end{equation}
In order to compare the best runtime for each method, the ED computations were done on the HPC while the NQS computations were conducted on the M1. In this work, the ED method used was the the (iterative) implicitly restarted Laczos algorithm \cite{Lanczos:1950zz}. Due to the number of variables at play (e.g. iteration depth), the expected growth of the ED compute time compared to the dimensions of the Hilbert space was not the focus of this analysis. Rather, we merely compare the required compute time for both methods. As can be seen in the figure, the ED computations (ranging from $\approx$ 0.2 seconds for $\jmax = 1$ to $\approx$ 233 seconds for $\jmax = 8$) far exceed their NQS counterparts (ranging from $\approx$ 8.8 seconds for $\jmax = 1$ to $\approx$ 536 seconds for $\jmax = 8$) in terms of compute time in low dimensional Hilbert spaces. However, for higher $\dim\hilb$, the gap in terms of efficiency starts to narrow. In fact, extrapolated data seem to indicate that the NQS approach excels in high $\dim\hilb$ cases as opposed to ED. Further, ED quickly become limited by the amount of required RAM to store the sparse (or dense) matrix. As is shown in later work \cite{Sahlmann:2024kat}, this is a hard limit which renders ED methods unusable while the NQS approach does not suffer from such an issue. 
\newline
Lastly, the compute time of the NQS approach is strongly affected by the topology of the network. Different architectures were implemented in this work with varying degrees of efficiency. However, this work did not focus on finding the most efficient architecture and thus, this is purposely left out. The door for reducing the compute time remains open and is to be explored in future work. Further, explicitly utilising parallelisation as well as GPUs is expected to reduce the compute time of the NQS approach drastically. 

\section*{References}
\begin{thebibliography}{}


\bibitem{Wald:1984rg}
R.~M.~Wald,
``General Relativity,''
Chicago Univ. Pr., 1984,
doi:10.7208/chicago/9780226870373.001.0001

\bibitem{Bondi:1947fta}
H.~Bondi,
``Spherically symmetrical models in general relativity,''
Mon. Not. Roy. Astron. Soc. \textbf{107}, 410-425 (1947)
doi:10.1093/mnras/107.5-6.410

\bibitem{Weinberg:1972kfs}
S.~Weinberg,
``Gravitation and Cosmology: Principles and Applications of the General Theory of Relativity,''
John Wiley and Sons, 1972,
ISBN 978-0-471-92567-5, 978-0-471-92567-5

\bibitem{Carroll:2004st}
S.~M.~Carroll,
``Spacetime and Geometry: An Introduction to General Relativity,''
Cambridge University Press, 2019,
ISBN 978-0-8053-8732-2, 978-1-108-48839-6, 978-1-108-77555-7
doi:10.1017/9781108770385

\bibitem{Carlip:1995zj}
S.~Carlip,
``Lectures on (2+1) dimensional gravity,''
J. Korean Phys. Soc. \textbf{28}, S447-S467 (1995)
[arXiv:gr-qc/9503024 [gr-qc]].

\bibitem{Carlip:1995qv}
S.~Carlip,
``The (2+1)-Dimensional black hole,''
Class. Quant. Grav. \textbf{12}, 2853-2880 (1995)
doi:10.1088/0264-9381/12/12/005
[arXiv:gr-qc/9506079 [gr-qc]].

\bibitem{Carlip:1998uc}
S.~Carlip,
``Quantum gravity in 2+1 dimensions,''
Cambridge University Press, 2003,
ISBN 978-0-521-54588-4, 978-0-511-82229-2
doi:10.1017/CBO9780511564192

\bibitem{Rovelli:1997yv}
C.~Rovelli,
``Loop quantum gravity,''
Living Rev. Rel. \textbf{1}, 1 (1998)
doi:10.12942/lrr-1998-1
[arXiv:gr-qc/9710008 [gr-qc]].

\bibitem{Thiemann:2001gmi}
T.~Thiemann,
``Modern canonical quantum general relativity,''
[arXiv:gr-qc/0110034 [gr-qc]].

\bibitem{Thiemann:2007pyv}
T.~Thiemann,
``Modern Canonical Quantum General Relativity,''
Cambridge University Press, 2007,
ISBN 978-0-511-75568-2, 978-0-521-84263-1
doi:10.1017/CBO9780511755682

\bibitem{Ashtekar:2004eh}
A.~Ashtekar and J.~Lewandowski,
``Background independent quantum gravity: A Status report,''
Class. Quant. Grav. \textbf{21}, R53 (2004)
doi:10.1088/0264-9381/21/15/R01
[arXiv:gr-qc/0404018 [gr-qc]].

\bibitem{Thiemann:1996aw}
T.~Thiemann,
``Quantum spin dynamics (QSD),''
Class. Quant. Grav. \textbf{15}, 839-873 (1998)
doi:10.1088/0264-9381/15/4/011
[arXiv:gr-qc/9606089 [gr-qc]].

\bibitem{Thiemann:1996av}
T.~Thiemann,
``Quantum spin dynamics (qsd). 2.,''
Class. Quant. Grav. \textbf{15}, 875-905 (1998)
doi:10.1088/0264-9381/15/4/012
[arXiv:gr-qc/9606090 [gr-qc]].

\bibitem{Reisenberger:1996pu}
M.~P.~Reisenberger and C.~Rovelli,
``'Sum over surfaces' form of loop quantum gravity,''
Phys. Rev. D \textbf{56}, 3490-3508 (1997)
doi:10.1103/PhysRevD.56.3490
[arXiv:gr-qc/9612035 [gr-qc]].

\bibitem{Varadarajan:2022dgg}
M.~Varadarajan,
``Anomaly free quantum dynamics for Euclidean LQG,''
doi:10.48550/arXiv.2205.10779
[arXiv:2205.10779 [gr-qc]].

\bibitem{Bojowald:2001xe}
M.~Bojowald,
``Absence of singularity in loop quantum cosmology,''
Phys. Rev. Lett. \textbf{86}, 5227-5230 (2001)
doi:10.1103/PhysRevLett.86.5227
[arXiv:gr-qc/0102069 [gr-qc]].

\bibitem{Ashtekar:2006wn}
A.~Ashtekar, T.~Pawlowski and P.~Singh,
``Quantum Nature of the Big Bang: Improved dynamics,''
Phys. Rev. D \textbf{74}, 084003 (2006)
doi:10.1103/PhysRevD.74.084003
[arXiv:gr-qc/0607039 [gr-qc]].

\bibitem{Bahr:2016hwc}
B.~Bahr and S.~Steinhaus,
``Numerical evidence for a phase transition in 4d spin foam quantum gravity,''
Phys. Rev. Lett. \textbf{117}, no.14, 141302 (2016)
doi:10.1103/PhysRevLett.117.141302
[arXiv:1605.07649 [gr-qc]].

\bibitem{Han:2020npv}
M.~Han, Z.~Huang, H.~Liu, D.~Qu and Y.~Wan,
``Spinfoam on a Lefschetz thimble: Markov chain Monte Carlo computation of a Lorentzian spinfoam propagator,''
Phys. Rev. D \textbf{103}, no.8, 084026 (2021)
doi:10.1103/PhysRevD.103.084026
[arXiv:2012.11515 [gr-qc]].

\bibitem{Cunningham:2020uco}
W.~J.~Cunningham, B.~Dittrich and S.~Steinhaus,
``Tensor Network Renormalization with Fusion Charges\textemdash{}Applications to 3D Lattice Gauge Theory,''
Universe \textbf{6} (2020) no.7, 97
doi:10.3390/universe6070097
[arXiv:2002.10472 [hep-th]].

\bibitem{Dona:2022yyn}
P.~Dona, M.~Han and H.~Liu,
``Spinfoams and High-Performance Computing,''
doi:10.1007/978-981-19-3079-9\_100-1
[arXiv:2212.14396 [gr-qc]].

\bibitem{Oda:1990rw}
I.~Oda and S.~Yahikozawa,
``Effective Actions of (2+1)-dimensional Gravity and Bf Theory,''
Class. Quant. Grav. \textbf{11}, 2653-2666 (1994)
doi:10.1088/0264-9381/11/11/008

\bibitem{Celada:2016jdt}
M.~Celada, D.~Gonz\'alez and M.~Montesinos,
``$BF$ gravity,''
Class. Quant. Grav. \textbf{33}, no.21, 213001 (2016)
doi:10.1088/0264-9381/33/21/213001
[arXiv:1610.02020 [gr-qc]].

\bibitem{Smolin:1992wj}
L.~Smolin,
``The G(Newton) ---\ensuremath{>} 0 limit of Euclidean quantum gravity,''
Class. Quant. Grav. \textbf{9}, 883-894 (1992)
doi:10.1088/0264-9381/9/4/007
[arXiv:hep-th/9202076 [hep-th]].

\bibitem{Bakhoda:2022rut}
S.~Bakhoda,
``The $U(1)^3$ Model of Euclidean Quantum Gravity,''

\bibitem{Giesel:2006uj}
K.~Giesel and T.~Thiemann,
``Algebraic Quantum Gravity (AQG). I. Conceptual Setup,''
Class. Quant. Grav. \textbf{24}, 2465-2498 (2007)
doi:10.1088/0264-9381/24/10/003
[arXiv:gr-qc/0607099 [gr-qc]].

\bibitem{Giesel:2006uk}
K.~Giesel and T.~Thiemann,
``Algebraic Quantum Gravity (AQG). II. Semiclassical Analysis,''
Class. Quant. Grav. \textbf{24}, 2499-2564 (2007)
doi:10.1088/0264-9381/24/10/004
[arXiv:gr-qc/0607100 [gr-qc]].

\bibitem{Baez:1999sr}
J.~C.~Baez,
``An Introduction to Spin Foam Models of $BF$ Theory and Quantum Gravity,''
Lect. Notes Phys. \textbf{543} (2000), 25-93
doi:10.1007/3-540-46552-9\_2
[arXiv:gr-qc/9905087 [gr-qc]].

\bibitem{Dittrich:2014wpa}
B.~Dittrich and M.~Geiller,
``A new vacuum for Loop Quantum Gravity,''
Class. Quant. Grav. \textbf{32} (2015) no.11, 112001
doi:10.1088/0264-9381/32/11/112001
[arXiv:1401.6441 [gr-qc]].

\bibitem{Drobinski:2017kfm}
P.~Drobi\'nski and J.~Lewandowski,
``Continuum approach to the BF vacuum: The U(1) case,''
Phys. Rev. D \textbf{96} (2017) no.12, 126011
doi:10.1103/PhysRevD.96.126011
[arXiv:1705.09836 [gr-qc]].

\bibitem{Bodendorfer:2016uat}
N.~Bodendorfer,
``An elementary introduction to loop quantum gravity,''
[arXiv:1607.05129 [gr-qc]].

\bibitem{Corichi:1997us}
A.~Corichi and K.~V.~Krasnov,
``Ambiguities in loop quantization: Area versus electric charge,''
Mod. Phys. Lett. A \textbf{13} (1998), 1339-1346
doi:10.1142/S0217732398001406
[arXiv:hep-th/9703177 [hep-th]].

\bibitem{Sahlmann:2002xv}
H.~Sahlmann,
``When do measures on the space of connections support the triad operators of loop quantum gravity?,''
J. Math. Phys. \textbf{52} (2011), 012503
doi:10.1063/1.3525706
[arXiv:gr-qc/0207112 [gr-qc]].

\bibitem{Ashtekar:1994mh}
A.~Ashtekar and J.~Lewandowski,
``Projective techniques and functional integration for gauge theories,''
J. Math. Phys. \textbf{36} (1995), 2170-2191
doi:10.1063/1.531037
[arXiv:gr-qc/9411046 [gr-qc]].

\bibitem{Sahlmann:2011xu}
H.~Sahlmann,
``Black hole horizons from within loop quantum gravity,''
Phys. Rev. D \textbf{84} (2011), 044049
doi:10.1103/PhysRevD.84.044049
[arXiv:1104.4691 [gr-qc]].

\bibitem{Thiemann:2021hpa}
T.~Thiemann and M.~Varadarajan,
``On Propagation in Loop Quantum Gravity,''
Universe \textbf{8} (2022) no.12, 615
doi:10.3390/universe8120615
[arXiv:2112.03992 [gr-qc]].

\bibitem{2012arXiv1209.1135G}
Gelca, R. \& Hamilton, A. 
``Classical theta functions from a quantum group perspective,''
arXiv E-prints. [arXiv:1209.1135 [[math.QA]]] (2012)

\bibitem{Dittrich:2018dvs}
B.~Dittrich,
``Cosmological constant from condensation of defect excitations,''
Universe \textbf{4} (2018) no.7, 81
doi:10.3390/universe4070081
[arXiv:1802.09439 [gr-qc]].


\bibitem{White:1992zz}
S.~R.~White,
``Density matrix formulation for quantum renormalization groups,''
Phys. Rev. Lett. \textbf{69} (1992), 2863-2866
doi:10.1103/PhysRevLett.69.2863

\bibitem{doi:10.1080/14789940801912366}
F. Verstraete, V. \& Cirac, J. 
``Matrix product states, projected entangled pair states, and variational renormalization group methods for quantum spin systems,''
Advances In Physics. \textbf{57}, 143-224 (2008)
https://doi.org/10.1080/14789940801912366

\bibitem{Schollwoeck:2010uqf}
U.~Schollwoeck,
``The density-matrix renormalization group in the age of matrix product states,''
Annals Phys. \textbf{326} (2011), 96-192
doi:10.1016/j.aop.2010.09.012
[arXiv:1008.3477 [cond-mat.str-el]].

\bibitem{Rommer:1997zz}
S.~Rommer and S.~Ostlund,
``Class of ansatz wave functions for one-dimensional spin systems and their relation to the density matrix renormalization group,''
Phys. Rev. B \textbf{55} (1997), 2164-2181
doi:10.1103/PhysRevB.55.2164
[arXiv:cond-mat/9606213 [cond-mat]].

\bibitem{Sandvik:1991axv}
A.~W.~Sandvik and J.~Kurkij\"arvi,
``Quantum Monte Carlo simulation method for spin systems,''
Phys. Rev. B \textbf{43} (1991) no.7, 5950
doi:10.1103/PhysRevB.43.5950

\bibitem{gubernatis_kawashima_werner_2016}
Gubernatis, J., Kawashima, N. \& Werner, P. 
``Quantum Monte Carlo Methods: Algorithms for Lattice Models,''
(Cambridge University Press, 2016)

\bibitem{1998PhLA..238..253P}
Prokof'ev, N., Svistunov, B. \& Tupitsyn, I. 
``“Worm” algorithm in quantum Monte Carlo simulations,''
Physics Letters A \textbf{238}, 253-257 (1998)

\bibitem{Troyer:2004ge}
M.~Troyer and U.~J.~Wiese,
``Computational complexity and fundamental limitations to fermionic quantum Monte Carlo simulations,''
Phys. Rev. Lett. \textbf{94} (2005), 170201
doi:10.1103/PhysRevLett.94.170201
[arXiv:cond-mat/0408370 [cond-mat]].

\bibitem{2011RvMP...83..863P}
Polkovnikov, A., Sengupta, K., Silva, A. \& Vengalattore, M. 
``Colloquium: Nonequilibrium dynamics of closed interacting quantum systems,''
Reviews Of Modern Physics \textbf{83}, 863-883 (2011)

\bibitem{doi:10.1142/1346}
Montorsi, A. 
``The Hubbard Model,'' 
(WORLD SCIENTIFIC, 1992)

\bibitem{thouless1972quantum}
Thouless, D. 
``The Quantum Mechanics of Many-body Systems,''
(Academic Press, 1972)

\bibitem{10.5555/523781}
Gurney, K. 
``An Introduction to Neural Networks,''
(Taylor \& Francis, Inc., 1997)

\bibitem{Carleo:2017nvk}
G.~Carleo and M.~Troyer,
``Solving the quantum many-body problem with artificial neural networks,''
Science \textbf{355} (2017) no.6325, 602-606
doi:10.1126/science.aag2302

\bibitem{10.1007/978-3-642-33275-3_2}
Fischer, A. \& Igel, C. 
``An Introduction to Restricted Boltzmann Machines,''
Progress In Pattern Recognition, Image Analysis, Computer Vision, And Applications, pp. 14-36 (2012)

\bibitem{PhysRevB.100.125124}
Choo, K., Neupert, T. \& Carleo, G.
``Two-dimensional frustrated $J_1 - J_2$ model studied with neural network quantum states,''
Phys. Rev. B \textbf{100}, 125124 (2019)
https://link.aps.org/doi/10.1103/PhysRevB.100.125124

\bibitem{PhysRevB.97.035116}
Cai, Z. \& Liu, J. 
``Approximating quantum many-body wave functions using artificial neural networks,''
Phys. Rev. B \textbf{97}, 035116 (2018)
https://link.aps.org/doi/10.1103/PhysRevB.97.035116

\bibitem{Choo_2020}
Choo, K., Mezzacapo, A. \& Carleo, G. 
``Fermionic neural-network states for ab-initio electronic structure,''
Nature Communications \textbf{11} (2020)
http://dx.doi.org/10.1038/s41467-020-15724-9

\bibitem{Nomura_2021}
Nomura, Y. 
``Helping restricted Boltzmann machines with quantum-state representation by restoring symmetry,''
Journal Of Physics: Condensed Matter \textbf{33}, 174003 (2021)
http://dx.doi.org/10.1088/1361-648X/abe268

\bibitem{PhysRevLett.128.090501}
Luo, D., Chen, Z., Carrasquilla, J. \& Clark, B.
``Autoregressive Neural Network for Simulating Open Quantum Systems via a Probabilistic Formulation,''
Phys. Rev. Lett. \textbf{128}, 090501 (2022)
https://link.aps.org/doi/10.1103/PhysRevLett.128.090501

\bibitem{Lin_2022}
Lin, S. \& Pollmann, F. 
``Scaling of Neural‐Network Quantum States for Time Evolution,''
Physica Status Solidi (b) \textbf{259} (2022)
http://dx.doi.org/10.1002/pssb.202100172

\bibitem{Wu:2023gic}
Y.~Wu, X.~Xu, D.~Poletti, Y.~Fan, C.~Guo and H.~Shang,
``A Real Neural Network State for Quantum Chemistry,''
Mathematics \textbf{11} (2023) no.6, 1417
doi:10.3390/math11061417
[arXiv:2301.03755 [quant-ph]].

\bibitem{Wu:2023kwx}
Y.~Wu, C.~Guo, Y.~Fan, P.~Zhou and H.~Shang,
``NNQS-Transformer: an Efficient and Scalable Neural Network Quantum States Approach for Ab initio Quantum Chemistry,''
doi:10.1145/3581784.3607061
[arXiv:2306.16705 [quant-ph]].

\bibitem{LIEBCHEN2007337}
Liebchen, C. \& Rizzi, R. 
``Classes of cycle bases,''
Discrete Applied Mathematics \textbf{155}, 337-355 (2007)
https://www.sciencedirect.com/science/article/pii/S0166218X06003052

\bibitem{Lanczos:1950zz}
C.~Lanczos,
``An iteration method for the solution of the eigenvalue problem of linear differential and integral operators,''
J. Res. Natl. Bur. Stand. B \textbf{45} (1950), 255-282
doi:10.6028/jres.045.026

\bibitem{OShea2015AnIT}
O’Shea, K. \& Nash, R. 
``An Introduction to Convolutional Neural Networks,''
ArXiv: \textbf{abs/1511.08458} (2015)
https://api.semanticscholar.org/CorpusID:9398408

\bibitem{Yamashita2018}
Yamashita, R., Nishio, M., Do, R. \& Togashi, K. 
``Convolutional neural networks: an overview and application in radiology,''
Insights Into Imaging \textbf{9}, 611-629 (2018)
https://doi.org/10.1007/s13244-018-0639-9

\bibitem{PhysRevLett.122.065301}
Levine, Y., Sharir, O., Cohen, N. \& Shashua, A. 
``Quantum Entanglement in Deep Learning Architectures,''
Phys. Rev. Lett. \textbf{122}, 065301 (2019)
https://link.aps.org/doi/10.1103/PhysRevLett.122.065301


\bibitem{7780459}
He, K., Zhang, X., Ren, S. \& Sun, J. 
``Deep Residual Learning for Image Recognition,''
2016 IEEE Conference On Computer Vision And Pattern Recognition (CVPR) pp. 770-778 (2016)
https://doi.ieeecomputersociety.org/10.1109/CVPR.2016.90

\bibitem{Luo:2022jzl}
D.~Luo, S.~Yuan, J.~Stokes and B.~K.~Clark,
``Gauge Equivariant Neural Networks for 2+1D U(1) Gauge Theory Simulations in Hamiltonian Formulation,''
[arXiv:2211.03198 [hep-lat]].

\bibitem{PhysRevLett.121.167204}
Choo, K., Carleo, G., Regnault, N. \& Neupert, T. 
``Symmetries and Many-Body Excitations with Neural-Network Quantum States,''
Phys. Rev. Lett. \textbf{121}, 167204 (2018)
https://link.aps.org/doi/10.1103/PhysRevLett.121.167204

\bibitem{PhysRevResearch.5.013216}
Luo, D., Chen, Z., Hu, K., Zhao, Z., Hur, V. \& Clark, B. 
``Gauge-invariant and anyonic-symmetric autoregressive neural network for quantum lattice models,''
Phys. Rev. Res. \textbf{5}, 013216 (2023)
https://link.aps.org/doi/10.1103/PhysRevResearch.5.013216

\bibitem{Dellaportas2003}
Dellaportas, P. \& Roberts, G.
``An Introduction to MCMC,''
Spatial Statistics And Computational Methods pp. 1-41 (2003)
https://doi.org/10.1007/978-0-387-21811-3$\_$1

\bibitem{Chib1995}
Chib, S. \& Greenberg, E. 
``Understanding the Metropolis-Hastings Algorithm,''
The American Statistician \textbf{49}, 327-335 (1995)
http://www.jstor.org/stable/2684568

\bibitem{doi:10.1080/10618600.1998.10474787}
Brooks, S. \& Gelman, A. 
``General Methods for Monitoring Convergence of Iterative Simulations,''
Journal Of Computational And Graphical Statistics \textbf{7}, 434-455 (1998)
https://www.tandfonline.com/doi/abs/10.1080/10618600.1998.10474787

\bibitem{Gelman2013BayesianDA}
Gelman, A., Carlin, J., Stern, H., Dunson, D., Vehtari, A. \& Rubin, D. 
``Bayesian data analysis, third edition,`` (2013)
https://api.semanticscholar.org/CorpusID:220694612

\bibitem{Vehtari_2021}
Vehtari, A., Gelman, A., Simpson, D., Carpenter, B. \& Bürkner, P. 
``Rank-Normalization, Folding, and Localization: An Improved $\hat{R}$ for Assessing Convergence of MCMC (with Discussion),''
Bayesian Analysis \textbf{16} (2021)
http://dx.doi.org/10.1214/20-BA1221


\bibitem{Sahlmann:2024kat}
H.~Sahlmann and W.~Sherif,
``Towards quantum gravity with neural networks: Solving quantum Hamilton constraints of 3d Euclidean gravity in the weak coupling limit,'' Class. Quantum Grav. (2024) (in press https://doi.org/10.1088/1361-6382/ad7c14) [arXiv:2405.00661 [gr-qc]], 





\bibitem{PhysRevResearch.2.023358}
Hibat-Allah, M., Ganahl, M., Hayward, L., Melko, R. \& Carrasquilla, J. 
``Recurrent neural network wave functions,''
Phys. Rev. Res. \textbf{2}, 023358 (2020)
https://link.aps.org/doi/10.1103/PhysRevResearch.2.023358

\bibitem{Brémaud2017}
Brémaud, P.
``Exact Sampling,''
Discrete Probability Models And Methods: Probability On Graphs And Trees, Markov Chains And Random Fields, Entropy And Coding pp. 509-534 (2017)
https://doi.org/10.1007/978-3-319-43476-6$\_$21

\bibitem{Thiemann:1997ru}
T.~Thiemann,
``QSD 4: (2+1) Euclidean quantum gravity as a model to test (3+1) Lorentzian quantum gravity,''
Class. Quant. Grav. \textbf{15} (1998), 1249-1280
doi:10.1088/0264-9381/15/5/011
[arXiv:gr-qc/9705018 [gr-qc]].

\bibitem{CARLEO2019100311}
Carleo, G., Choo, K., Hofmann, D., Smith, J., Westerhout, T., Alet, F., Davis, E., Efthymiou, S., Glasser, I., Lin, S., Mauri, M., Mazzola, G., Mendl, C., Van Nieuwenburg, E., O’Reilly, O., Théveniaut, H., Torlai, G., Vicentini, F. \& Wietek, A. 
``NetKet: A machine learning toolkit for many-body quantum systems,''
SoftwareX \textbf{10} pp. 100311 (2019)
https://www.sciencedirect.com/science/article/pii/S2352711019300974

\bibitem{Sorella_2007}
Sorella, S., Casula, M. \& Rocca, D. ``Weak binding between two aromatic rings: Feeling the van der Waals attraction by quantum Monte Carlo methods," The Journal Of Chemical Physics \textbf{127} (2007) http://dx.doi.org/10.1063/1.2746035

\bibitem{Zhang_2023}
Zhang, S., Wan, Z. \& Yao, H. ``Automatic differentiable Monte Carlo: Theory and application," Physical Review Research \textbf{5} (2023), http://dx.doi.org/10.1103/PhysRevResearch.5.033041


\bibitem{kingma2017adammethodstochasticoptimization}
Kingma, D. \& Ba, J. ``Adam: A Method for Stochastic Optimization," (2017), https://arxiv.org/abs/1412.6980


\bibitem{jax2018github}
Bradbury, J., Frostig, R., Hawkins, P., Johnson, M., Leary, C., Maclaurin, D., Necula, G., Paszke, A., VanderPlas, J., Wanderman-Milne, S. \& Zhang, Q. ``JAX: composable transformations of Python+NumPy programs,"  (2018), http://github.com/google/jax

\bibitem{50530}
Sabne, A. ``XLA : Compiling Machine Learning for Peak Performance,"  (2020)


\end{thebibliography}
\end{document}